\documentclass{l4dc2025}

\title[Exploiting inter-agent coupling information for efficient rl of cooperative lqr]{Exploiting inter-agent coupling information for efficient reinforcement learning of cooperative LQR}
\usepackage{times}
\author{%
 \Name{Shahbaz P Qadri Syed} \Email{shahbaz\_qadri.syed@okstate.edu}\\
 \Name{He Bai} \Email{he.bai@okstate.edu}\\
 \addr Mechanical and Aerospace Engineering, Oklahoma State University, USA.%
}

\usepackage{amsfonts,amsmath,amssymb,cases,algorithm,enumerate,tikz}
\usetikzlibrary{bayesnet}
\usepackage[ruled]{algorithm2e}
\usepackage{multicol, booktabs,graphicx,subcaption}
\newcommand{\I}{\mathcal{I}}
\newcommand{\R}{\mathbb{R}}
\newcommand{\bphi}{\mathbf{\Phi}}
\newcommand{\bpsi}{\mathbf{\Psi}}
\newcommand{\bxi}{\mathbf{\Xi}}

\newcommand{\T}{\intercal}

\newcommand{\mS}{\mathcal{S}}
\newcommand{\mZ}{\mathcal{Z}}

\newcommand{\mc}[1]{\mathcal{#1}}
\newcommand{\mb}[1]{\mathbb{#1}}
\newcommand{\mbf}[1]{\mathbf{#1}}
\newcommand{\norm}[1]{\left\lVert #1 \right\rVert}
\newcommand{\lrb}[1]{\left[#1\right]}
\newcommand{\lrp}[1]{\left(#1\right)}
\newcommand{\lrs}[1]{\{#1\}}
\newcommand{\trace}[1]{\text{Tr}\left(#1\right)}
\newtheorem{lm}{\textbf{Lemma}}[section]

\newtheorem{thm}{\textbf{Theorem}}[section]
\newtheorem{prop}{\textbf{Proposition}}[section]

\begin{document}

\maketitle

\begin{abstract}%
 Developing scalable and efficient reinforcement learning algorithms for cooperative multi-agent control has received significant attention over the past years. Existing literature has proposed \textit{inexact} decompositions of local Q-functions based on empirical information structures between the agents. In this paper, we exploit inter-agent coupling information and propose a systematic approach to \textit{exactly} decompose the local Q-function of each agent. We develop an approximate least square policy iteration algorithm based on the proposed decomposition and identify two architectures to learn the local Q-function for each agent. We establish that the worst-case sample complexity of the decomposition is equal to the centralized case and derive necessary and sufficient graphical conditions on the inter-agent couplings to achieve  better sample efficiency. We demonstrate the improved sample efficiency and computational efficiency on numerical examples.
\end{abstract}


\section{Introduction}

Owing to the recent advancements in multi-agent reinforcement learning (MARL), there is an increasing interest to investigate MARL as a solution concept for model-free optimal control of uncertain cooperative multi-agent systems (MAS). However, scaling from single agent RL to MARL poses two main challenges: first, a centralized approach for learning is limited by the curse of dimensionality because of the combinatorial joint state-control space of the agents. Several recent works have attempted to address the scalability issue by approximation or relaxed assumptions during learning (e.g.,~\cite{lowe2017, foerster2017,zhang2018fully,guestrin2001,sunehag2017,rashid2018,qu2020scalable,jing2024distributed}).

  Second, each agent typically interacts with only a subset of agents resulting in a partial information structure and \textit{non-stationarity} during learning. In addition, the problem of finding globally optimal controllers with information constraints is known to be NP-hard even in a model-based setting~(\cite{witsenhausen1968counterexample,blondel2000survey,papadimitriou1986intractable,ye2022sample}). 
  Due to these challenges, the effectiveness of MARL for optimal control of uncertain MAS is not yet fully explored even in the linear quadratic regulator (LQR) setting, where the agent dynamics are linear time-invariant and the cost is quadratic.

  In this paper,  we study a cooperative LQR problem and propose a systematic approach to decompose the Q-function of each agent and address the issues of information constraint and scalability. We identify \textit{exact} decompositions of the individual Q-function and the gradient of the global Q-function with respect to each agent's control parameters, given the knowledge of state-, cost-, and observation couplings between the agents. Leveraging inter-agent couplings to develop efficient RL algorithms for MAS is of great interest recently. For instance, \cite{qu2020scalable} develop a Q-function approximation strategy based on the spatially exponentially decaying (SED) property~(\cite{bamieh2002distributed,motee2008optimal,gamarnik2014correlation,qu2020scalable}) which exploits the structure of local interactions in state transitions assuming decentralized observations and cost. 
   \cite{alemzadeh2019distributed} consider a distributed Q-learning approach for dynamically decoupled agents with coupled costs. For an LQR problem,~\cite{alemzadeh2019distributed} show that the controller of each subsystem asymptotically converges to the optimal controller.  Other distributed RL methods~(e.g., \cite{kar2013QD,zhang2018fully,macua2018diff,zhang2020cooperative,zhang2021finite,li2023f2a2,qu2019value}) have been proposed to employ a consensus algorithm to estimate the global cost with local information from neighbors. The performance of such algorithms depends on convergence of the consensus and the estimation quality of the global cost. Another relevant class of methods are value function factorizatio based methods~(\cite{koller1999,guestrin2001,sunehag2017,rashid2018,son2019}) and coordination graph (CG) based methods~(\cite{guestrin2001, guestrin2002, kok2006}), which aim to \textit{approximately} factorize the global Q-function using the individual Q-functions of all the agents. 
   However, the existing methods typically consider one or two types of inter-agent couplings and perform approximate decompositions of the Q-function. In contrast, in this paper we consider that local interactions between the agents occur in the form of state, observation and cost couplings and present an approach to deduce an exact decomposition of the individual Q-function.

   The formulation presented in this paper is closely related to~\cite{jing2024distributed} that also incorporates the three types of inter-agent couplings and proposes a distributed RL framework based on zeroth order optimization. However, there are a few key differences. First,~\cite{jing2024distributed} consider a value function based approach that depends on the global state and action whereas we identify an exact decomposition of the Q-function which is more closely related to the actor-critic methods~(\cite{konda1999actor}) in RL. Second,~\cite{jing2024distributed} introduce a new graph based on the inter-agent couplings called \textit{learning graph} which represents the information flow during learning. This learning graph is equivalent to the gradient dependency graph introduced in this paper. However, we derive this graph independently based on the decomposition of the Q-function.
  
  LQR is a popular benchmark in RL to analyze the performance and limitations for various algorithms~(\cite{recht2019tour}). Although the model-free LQR setting (for both structured and unstructured) is well studied in single agent RL~(e.g.,~\cite{bradtke1996linear,lagoudakis2003least,krauth2019finite,park2020structured}), it is less explored in the multi-agent setting.  For some recent results in the networked LQR setting, see~\cite{li2021distributed, jing2021learning,jing2021model,alemzadeh2019distributed,zhang2023optimal,olsson2024scalable}. Another contribution of this paper is that we analyze the sample complexity and the estimation error in Q-function parameters for the cooperative LQR problem using the least square policy iteration (LSPI) framework proposed in \cite{lagoudakis2003least}. In particular, we propose a \textit{multi-agent LSPI} (MALSPI) algorithm based on the proposed Q-function decomposition and identify two architectures (\textit{direct} and \textit{indirect}) to learn the local Q-function for each agent. We establish that the worst case sample complexity of the direct case is equal to the centralized case, and the worst case sample complexity of the indirect case is equal to the direct case. We also derive the necessary and sufficient graphical conditions on the inter-agent couplings for improved sample complexity of the decomposition. Finally, we validate the sample complexity results using a numerical example.

The rest of the paper is outlined as follows. Section~\ref{sec:probform} presents the formulation of the cooperative LQR problem. Section~\ref{sec:decompQ} introduces the decomposition of the local Q-function and the decomposition of the gradient of the global Q-function. Section~\ref{sec:alg} describes the MALSPI algorithm, and Section~\ref{sec:theory} states the main sample complexity results of the algorithm. Simulation results are presented in Section~\ref{sec:sim}. Section~\ref{sec:conc} concludes the paper.
  
\textit{Notation:}  We denote a Gaussian distribution with mean $\mu$ and covariance $\Sigma$ as $\mc{N}(\mu,\Sigma).$ Let $\mb{I}_n$ and $\mbf{0}_n$ denote an identity matrix and a zero vector of size $n$ respectively. For an $n$-dimensional symmetric matrix $M$, let $\text{svec}(M)\in \R^{\frac{n(n+1)}{2}}$ be the vector of the upper triangular entries of $M$ such that $||M||^2_F = \langle \text{svec}(M),\text{svec}(M)\rangle$ and $\text{smat}(\cdot)$ represent the inverse operation of $\text{svec}(\cdot)$. Let $\mc{L}(X,Y)$ be the analytical solution of the discrete-time Lyapunov equation $\mc{P} = X \mc{P} X^\T + Y.$ 
\section{Formulation of the cooperative LQR problem}\label{sec:probform}
Consider a multi-agent system comprising $N$ agents with linear time-invariant dynamics and quadratic costs. Let $\mc{V}=\{1,\cdots,N\}$. We assume a generic setting, where each agent's dynamics and cost may depend on other agents' states and controls and its observations are the states of a subset of all the agents. Specifically,  the dynamics of agent $i \in \mc{V}$ is given by
\begin{align}
    x_i({t+1})&= \sum_{j\in \I^i_S} A_{ij} x_j(t) + \sum_{j\in \I^i_S} B_{ij} u_j(t) + w_i(t),
    \label{eq:inddyn}
\end{align}
where $x_i(t)\in\mathbb{R}^{n_x}$, $u_i(t)\in\mathbb{R}^{n_u}$, $w_i(t) \sim \mathcal{N}(0,\sigma^2_w~\mb{I}_{n_x})$ is the process noise, and the set $\mathcal{I}^i_S$ contains the indices of the agents impacting the dynamics of agent $i$. For the ease of exposition, we have assumed that the dimensions of $x_i(t)$ and $u_i(t)$ are homogeneous across all the agents. However, the results in this paper can be easily extended to heterogeneous state and control dimensions. Let $x_{\I^i_S}(t) = [\cdots~x^\T_j(t)~\cdots]^\T$ and $u_{\I^i_S}(t)=[\cdots~u^\T_j(t)~\cdots]^\T$, $\forall~j\in\I^i_S$. 

The quadratic cost incurred by agent $i \in \mc{V}$ at time $t$ is given by
\begin{align}
    c_i(x_{\I^i_C}(t),u_{\I^i_C}(t)) &= (x_{\I^i_C}(t))^\T S_i (x_{\I^i_C}(t)) +(u_{\I^i_C}(t))^\T R_i (u_{\I^i_C}(t)), \label{eq:indcost}
\end{align}
where the set $\mathcal{I}^i_C$ contains the indices of the agents impacting the cost of agent $i$, $x_{\I^i_C}(t) = [\cdots~x^\T_j(t)~\cdots]^\intercal$, and $u_{\I^i_C}(t)=[\cdots~u^\T_j(t)~\cdots]^\intercal,$ $\forall~j\in\I^i_C.$ We assume that $S_i$ is a positive semi-definite matrix and $R_i$ is a positive {definite} matrix. 

Agent $i$ observes the states of the agents in a set $\I^i_O$ and prescribes a structured linear control \begin{align}
    \pi_{i}(x_{\I^i_O}(t)) := u_i(t) &= K_i x_{\I^i_O}(t), 
    \label{eq:indcontrol}
\end{align}
where $K_i = [\cdots~K_{ij}~\cdots],~x_{\I^i_O}(t) = [\cdots~x^\T_j(t)~\cdots]^\T,$ $\forall~j\in \I^i_O.$

The agent couplings in state dynamics, costs, and observations are captured by the index sets $\mathcal{I}^i_S$, $\mathcal{I}^i_C$, and $\mathcal{I}^i_O$, $\forall i$. We assume that these sets are time-invariant and further define a \textit{state graph} $\mc{G}_S = \{\mc{V}, \mc{E}_S\},$  an \textit{observation graph} $\mc{G}_O = \{\mc{V}, \mc{E}_O\},$  and  a \textit{cost graph} $\mc{G}_C = \{\mc{V}, \mc{E}_C\}$. An edge $(i,j) \in \mc{E}_S$, if $i\in \mathcal{I}^j_S$. Similarly, an edge $(i,j) \in \mc{E}_O$ if $i\in\mathcal{I}^j_O$ and an edge $(i,j) \in \mc{E}_C$ if $i\in\mathcal{I}_C^j$. We now present the combined state dynamics, cost, and controller. Let $x(t) = [x^\T_1(t)~~\cdots~~x^\T_N(t)]^\T \in \R^{N n_x},~u(t) = [u^\T_1(t)~~\cdots~~u^\T_N(t)]^\T \in \R^{N n_u}$, and $w(t) = [w^\T_1(t)~~\cdots~~w^\T_N(t)]^\T \in \R^{N n_x}$. The global state of the multi-agent system, $x(t)$, evolves according to
\begin{align}
  x(t+1) =  A x(t) + B u(t) + w(t), 
  \label{eq:dynamics}
\end{align}
where $A\in \R^{Nn_x \times N n_x}$ and $B\in \R^{N n_x \times N n_u}$ are \textit{unknown} system matrices depending on $A_{ij}$ and $B_{ij}$ in~\eqref{eq:inddyn}, respectively. Observe that $~\forall~i\in \mc{V},~\forall~j \in \mc{V}\setminus \I^i_S$, $A_{ij} = \mbf{0}_{n_x\times n_x}$, and $B_{ij} = \mbf{0}_{n_x\times n_u}$.  

In our cooperative LQR problem, the global cost $C(x(t), u(t))$ is given by the average of the individual cost functions, i.e.,
$C(x(t),u(t)) = \frac{1}{N}\sum_{i = 1}^N c_i(x_{\I^i_C}(t),u_{\I^i_C}(t))$,  which is rewritten as 
\begin{align}
    C(x(t),u(t)) = x^\T(t) S x(t) + u^\T(t) R u(t),\quad \forall~t,\label{eq:cost}
\end{align}
where $S$ and $R$ are  cost matrices depending on $S_i$ and $R_i$ in~\eqref{eq:indcost}. Note that $\forall~i\in \mathcal{V},~\forall~j\in \mathcal{V}\setminus \I^i_C,$ $S_{ij} = \mbf{0}_{n_x\times n_x},$ and $R_{ij} = \mbf{0}_{n_u\times n_u}$. The prescribed structured static linear controllers~\eqref{eq:indcontrol} represented in a compact form is given by $\pi(x(t)):= u(t) = K x(t),$ 
where the control gain $K$ depends on $\mc{G}_O$ such that $\forall~i\in \mc{V},~\forall~j\in \mc{V}\setminus \I^i_O$, $K_{ij} = \mbf{0}_{n_u \times n_x}.$

We seek to optimize the control $u(t)$ to minimize the expected long term cost. 
The \textit{global} Q-function for the cooperative LQR problem is defined as $Q(x,u) = \mathbb{E}[\sum_{t=0}^\infty C(x(t),u(t))| x(0)\! =\! x,$ $ u(0)\! =\! u]$, where the expectation is taken over the state and control distributions, respectively. For agent $i\!\in\!\mc{V}$, define a \textit{local} $Q$ function as  $Q_i(x,u) = \mathbb{E}\left[\sum_{t=0}^\infty c_i(x_{\mathcal{I}^i_C}(t), u_{\mathcal{I}^i_C}(t)) | x(0) \!=\! x, u(0)\! = \!u\right]$, which satisfies $Q(x,u) = \sum_{i=1}^N Q_i(x,u)$. 

The objective of an infinite horizon average cost cooperative LQR problem is to compute an optimal controller 
$u^*(t)=K_* x(t),$
that minimizes $$J(x(0),u(0)) = \underset{{K}}{\text{min }} \underset{T\rightarrow \infty}{\text{ lim}} \mb{E}_{x(0), u(0)} \left[Q(x(0),u(0))\right] = \underset{{K}}{\text{min }} \underset{T\rightarrow \infty}{\text{ lim}} \mb{E}\left[ \sum_{t=0}^T x^\T(t) S x(t) + u^\T(t) R u(t)\right],$$
when the system starts from a global state $x(0)$, executes a global control $u(0)$ and follows a policy {$\pi(\cdot)$} thereafter. { Given a global state $x(t)$, the global controller 
$u(t)$ and the collection of local controllers 
$\lrs{u_i(t)}_{i=1}^N$can be transformed into each other.} Therefore, our cooperative LQR problem can be concisely expressed as
\begin{eqnarray}
    \underset{K}{\text{min }}& \mathbb{E}_{x(0), u(0)} ~\lrb{Q(x(0), u(0))|\mc{G}_S,\mc{G}_O,\mc{G}_C}\nonumber\\
    \text{subject to }& x_i(t+1) \sim \mc{N} \lrp{ \sum_{j\in \I^i_S} \lrb{A_{ij} x_j(t) + B_{ij} u_j(t)},  \sigma^2_w~\mb{I}_{n_x}},\nonumber\\
    & u_i(t) = K_i x_{\I^i_O}(t),~\forall~i\in \mc{V}.
    \label{eq:nmasopt}
\end{eqnarray}

For general linear time-invariant systems, there is no known tractable algorithm for computing optimal $K_i$~(\cite{rotkowitz2005characterization}). Moreover,~\cite{blondel2000survey} showed that the problem of finding stabilizing decentralized output feedback is NP-complete. Tractable algorithms  that guarantee a global optimal controller have been developed for specific information structures, such as partially nested information structures~(\cite{ho1972team}), quadratic invariance~(\cite{rotkowitz2005characterization}), partially ordered sets~(\cite{shahposets2013}), and decoupled control cost~(\cite{kashyap2023guaranteed}). In this work, we propose an algorithm that aims at solving the general non-convex optimization problem~\eqref{eq:nmasopt} (with a static output feedback) using a policy iteration algorithm. Our main goal is to demonstrate improved sample complexity in estimating $Q$ functions by recognizing structures in agent couplings. Convergence properties of the proposed algorithm will be investigated in the future work. 

The system matrices $A$, $B$, and the cost matrices $S$, $R$ are unknown in the model-free setting. However, their sparsity patterns may be known based on the inter-agent couplings (e.g., through physical and network couplings). Thus, we assume that the agent coupling graphs $\mc{G}_S$, $\mc{G}_O$, and $\mc{G}_C$ are available and develop a systematic procedure to exploit the interplay between the system sparsity and the coupling graphs to improve the sample complexity of the cooperative LQR problem. Specifically, we decompose the $Q_i$ function given $\mc{G}_S$, $\mc{G}_O$, and $\mc{G}_C$ (Section~\ref{sec:decompQ}), from which a Q-learning algorithm is proposed and analyzed (Section~\ref{sec:alg} and~\ref{sec:theory}). 
 
\section{Decomposition of the Q-function}\label{sec:decompQ}
Due to the cooperative average quadratic cost, solving~\eqref{eq:nmasopt} is equivalent to minimizing the individual expected $Q_i(\cdot)$, $\forall~i\in\mc{V}$ which is typically assumed to be dependent on the global state and global control in the literature (see e.g.,~\cite{jing2024distributed},~\cite{lowe2017},~\cite{zhang2018fully}). Such a dependency incurs a combinatorial state-control dimension that grows exponentially with the number of agents leading to the curse of dimensionality.  
Since $Q_i(x(t),u(t)) = c_i(x_{\mathcal{I}^i_C}(t), u_{\mathcal{I}^i_C}(t)) + \mb{E}\lrb{Q_i(x(t+1),u(t+1))}$, $Q_i(\cdot)$ depends only on $\I^i_C$ and the subset of agents required to propagate the states and controls of the agents in $\I^i_C$ through time. That is, $Q_i(\cdot)$ depends on a closed subset (under state transition and control) of agents required to compute $\I^i_C$ at each time instant. We formalize this notion in Lemma~\ref{lm:rtc}. Let $\mc{G}_{SO} = \mc{G}_S\cup \mc{G}_O$ and  define $\mc{R}^i_{SO} = \lrs{j\in \mc{V}|j \xrightarrow{\mc{E}_{SO}}i}\cup \lrs{i}$ as the reachability set of $i$ in the $\mc{G}_{SO}.$ 
Lemma~\ref{lm:rtc}\footnote{The proofs of Lemma~\ref{lm:Qset} and  the other theoretical results can be found in the appendix.} shows that $\forall i\in \mc{V}$, the set $\I^i_Q \triangleq \lrs{j\in \mc{R}^k_{SO}| \forall~k\in \I^i_C},$ is closed under state transition and control generation.
\begin{lm}\label{lm:Qset}
    For any $i,j,k \in \mathcal{V}$, if $j\in \I^i_Q$, then for any $k\in \mathcal{R}^j_{SO}$, $k\in \I^i_Q$.
    \label{lm:rtc}
    \end{lm}
    
Theorem~\ref{thm:Qsetdecomp} below establishes that $Q_i(\cdot)$ depends \textit{only} on the states, and controls of the agents in $\I^i_Q$. Thus, we refer to $\I^i_Q$ as the \textit{value dependence set} of agent $i$ and the  graph such that the in-neighbors of node $i$ correspond to $\I^i_Q$ as the \textit{value dependency graph} denoted by $\mc{G}_{Q} = \{\mc{V},\mc{E}_{Q}\}.$
\begin{thm}[Value decomposition theorem]
    Let $\mc{R}^i_{SO} = \lrs{j\in \mc{V}|j \xrightarrow{\mc{E}_{SO}}i}\cup \lrs{i}$. Then, $\forall$ $i\in \mathcal{V}$, $Q_i(x(t),u(t)) = Q_i(x_{\I^i_Q}(t), u_{\I^i_Q}(t)), $ \text{where } $\I^i_Q = \lrs{j\in \mc{R}^k_{SO}| \forall~k\in \I^i_C}.$
    \label{thm:Qsetdecomp}
\end{thm}

Solving the optimization in~\eqref{eq:nmasopt} requires each agent $i\in\mc{V}$ to compute the gradient $\nabla_{K_i}J(\cdot)$ which is again a global computation. Theorem~\ref{thm:Qsetdecomp} implies that $K_i$ affects $Q_j$ only if $i \in \I^j_{Q}$, from which we pursue a decomposition of the gradient of the global objective w.r.t. (with respect to) $K_i.$  Define the \textit{gradient dependency} graph $\mc{G}_{\text{GD}} = \lrs{\mc{V}, \mc{E}^\intercal_{Q}}$ and the corresponding index set $\I^i_{\text{GD}} =\lrs{j\in \mc{V}|(j,i)\in \mc{E}^\intercal_{Q}}$.  Theorem~\ref{thm:graddecomp} below shows that the gradient of $Q(\cdot)$ w.r.t. $K_i$ can be decomposed as the  sum of the gradients of $Q_j$ w.r.t. $K_i$, $\forall$ $j \in \I^i_{\text{GD}}.$
\begin{thm}[Gradient decomposition theorem]
For the cooperative LQR problem defined in Section~\ref{sec:probform}, we have $\forall~t\geq 0$, $\forall$ $i \in \mathcal{V}$, \\
\vspace{-1em}
        $$\nabla_{K_i} Q(x(t),u(t)) = \nabla_{K_i} \left(\sum_{j \in \mathcal{I}_{\text{GD}}^i} Q_j(x_{\I^j_Q}(t),u_{\I^j_Q}(t))\right).$$
\vspace{-1.5em}
    \label{thm:graddecomp}
\end{thm}

Owing to Theorem~\ref{thm:graddecomp}, we further define $\I^i_{\widehat{Q}} \triangleq \lrs{k \in \mc{V}| k \in \bigcup_{j\in \I^i_{\text{GD}}} \I^j_{Q}}$ and $\widehat{Q}_i(x_{\I^i_{\widehat{Q}}},u_{\I^i_{\widehat{Q}}})\triangleq\sum_{j\in \I^i_{\text{GD}}} Q_j(x_{\I^j_Q},u_{\I^j_Q})$.  
The multi-agent deterministic policy gradient theorem~\cite{lowe2017} is an extension of the deterministic policy gradient theorem~\cite{silver14} to an MARL setting that gives the gradient of the objective $J(\cdot)$ w.r.t. the policy parameters of agent $i$ as $\nabla_{K_i} J(x,u) = \mathbb{E}\left[\nabla_{u_i} Q(x,u) \nabla_{K_i} u_i\right].$
Applying Theorem~\ref{thm:Qsetdecomp} and~\ref{thm:graddecomp} 
yields our cooperative deterministic policy gradient: $\forall~i\in\mc{V}$, 
\vspace{-0.7em}
\begin{align}
  \nabla_{K_i} J(x,u) = \mathbb{E}\left[ \nabla_{K_i} u_i \cdot \nabla_{u_i} \widehat{Q}_i(x_{\I^i_{\widehat{Q}}},u_{\I^i_{\widehat{Q}}})\right]
   =\mathbb{E}\left[ \nabla_{K_i} u_i \cdot \nabla_{u_i} \sum_{j\in \I^i_{\text{GD}}} Q_j(x_{\I^j_Q},u_{\I^j_Q})\right]. 
  \label{eq:strpgt}
\end{align}

In the remainder of the paper, we focus on two architectures for 
computing the policy gradient: 1) the \textit{direct} case (using the first equality in~\eqref{eq:strpgt}), where agent $i$, $\forall i\in \mc{V}$, directly estimates $\widehat{Q}_i(x_{\I^i_{\widehat{Q}}},u_{\I^i_{\widehat{Q}}})$ as a quadratic function of $x_{\I^i_{\widehat{Q}}},u_{\I^i_{\widehat{Q}}}$; 2) the \textit{indirect} case (using the second equality in~\eqref{eq:strpgt}), where agent $i$, $\forall i\in \mc{V}$, estimates $Q_i$ as a quadratic function of $x_{\I^i_{{Q}}},u_{\I^i_{{Q}}}$ and communicates with the agents in $\mc{G}_{\text{GD}}$ to compute~$\sum_{j\in \I^i_{\text{GD}}} Q_j(x_{\I^j_Q},u_{\I^j_Q})$. 
\section{Multi-agent structured least square policy iteration}\label{sec:alg}
In this section, we present the \textit{multi-agent least square policy iteration (MALSPI)} to solve~\eqref{eq:nmasopt} in a model-free setting. This algorithm extends the least square policy iteration~(\cite{lagoudakis2003least})--which is well understood both theoretically and empirically in a single agent setting--to MAS utilizing the decomposition in Theorem~\ref{thm:Qsetdecomp} and~\ref{thm:graddecomp}. The proposed MALSPI algorithm is an off-policy algorithm that employs a {shared experience buffer}. In each iteration, a trajectory rollout of $T$ samples is collected using a stabilizing policy $K_{\text{play}}(\neq K${, in general}$)$. We assume that $K_{\text{play}}$ is either known (e.g., for open-loop stable dynamics) or learned in a model-free setting (e.g.,~\cite{jing2021learning}).  Then each agent performs consecutive policy evaluation and policy improvement steps in parallel. The proposed MALSPI algorithm for the \textit{direct} case is summarized in Algorithm~\ref{alg:malspi}.

\begin{algorithm}[htpb]
 \caption{Multi-agent Least Square Policy Iteration (MALSPI) - the direct case}
 \label{alg:malspi} 
Input: Initial stabilizing controller $K_0$,  number of policy iterations $n$,  length of trajectory rollout $T$,  exploration noise variance $\sigma^2_\eta$, lower eigenvalue bound $\zeta$ , learning rate parameter $\alpha$, direct VD set $\I^i_{\widehat{Q}},$ $\forall~i\in \mc{V}$, global initial state mean $x_0$, and covariance $\Sigma_0$.\\
\For{$l = 0,\cdots,n$}{
Starting from a global state $x^{(l)}(0) \sim \mc{N}(x_0,\Sigma_0)$, and using an arbitrary policy $u^{(l)}(t) = K_0 x^{(l)}(t) \!+\! \eta^{(l)}_t,~\eta^{(l)}(t)\sim \mc{N}(\mbf{0},\sigma^2_\eta \mb{I}_{Nn_u})$, collect sample global trajectory $\mc{D}^{(l)} = \{x^{(l)}(t), u^{(l)}(t), x^{(l)}(t\!+\!1)\}_{t=1}^T.$\\
\For{$i = 1,2,\cdots, N$ (in parallel)}{
Query $\mc{D}^{(l)}_{\I^i_{\widehat{Q}}} = \lrs{x^{(l)}_{\I^i_{\widehat{Q}}}(t),u^{(l)}_{\I^i_{\widehat{Q}}}(t),x^{(l)}_{\I^i_{\widehat{Q}}}(t\!+\!1)}_{t=0}^T$ from $\mc{D}^{(l)}$\\
$\hat{q}_i \leftarrow \text{LSTDQ} \lrp{\mc{D}^{(l)}_{\I^i_{\widehat{Q}}}, K^{(l)}_{\I^i_{\widehat{Q}}}}$ [\eqref{eq:eiv_lse}]; $\widehat{Q}_i = \text{psd\_proj}_\zeta\lrp{\text{smat}(\hat{q}_i)}$\\
$K^{(l\!+\!1)}_i \leftarrow K^{(l)}_i - 2\alpha \mathbb{E}\left[ \mb{J}_i\widehat{Q}_i\begin{bmatrix}
        x_{\I^i_{\widehat{Q}}}\\u_{\I^i_{\widehat{Q}}}
    \end{bmatrix}  x_{\I^i_O}^\T\right]$ [\eqref{eq:detpgtexplicit}]
}
}
\end{algorithm}

\noindent
\textbf{Policy evaluation.}
We discuss the policy evaluation step in the direct case. The analysis extends to the indirect case in a straightforward manner. We assume that each agent $i\in \mc{V}$ has access to the evaluation policies of the agents in its $\I^i_{\widehat{Q}}.$ Given $\bigcup_{j\in \I^i_{\widehat{Q}}} K_j$, agent $i\in \mc{V}$ estimates its corresponding Q-function using least squares temporal difference learning for Q-functions (LSTDQ)~(\cite{lagoudakis2003least}). To simplify the notation, define a projection operator $P^n_{\mc{S}_1, \mc{S}_2} \in \R^{n|\mc{S}_1|\times n|\mc{S}_2|}$ w.r.t. ordered subsets $\mc{S}_1\subseteq \mc{S}_2 \subseteq \mc{V}$ such that $\forall~i\in \mc{S}_1,~j\in\mc{S}_2$, $P_{ij}\in \R^{n \times n}$ satisfies $P_{ij} = \begin{cases}
   \mb{I}_n, & \text{ if $j\in \mc{S}_1$}\\
   \mbf{0}_{n\times n},& \text{ otherwise}
\end{cases}.$ Since  $\forall~j\in \I^i_{\widehat{Q}}$, $\I^j_O\subseteq \I^i_{\widehat{Q}}$ and $\I^j_S \subseteq \I^i_{\widehat{Q}}$ by Lemma~\ref{lm:rtc}, let $K_{\I^i_{\widehat{Q}}} = [\cdots~(K_j P^{n_x}_{\I^j_O,\I^i_{\widehat{Q}}})^\intercal~\cdots]^\T, A_{\I^i_{\widehat{Q}}} = [\cdots~(A_j P^{n_x}_{\I^j_S,\I^i_{\widehat{Q}}})^\intercal~\cdots]^\T$, $B_{\I^i_{\widehat{Q}}} = [\cdots~(B_j P^{n_u}_{\I^j_S,\I^i_{\widehat{Q}}})^\intercal~\cdots]^\T.$
It then follows that $\forall~t \geq 0$, $\forall~i\in \mc{V}$,  $\forall j\in \I^i_{\widehat{Q}}$,
\begin{align}
   x_{\I^i_{\widehat{Q}}}(t+1) = A_{\I^i_{\widehat{Q}} }x_{\I^i_{\widehat{Q}}}(t) + B_{\I^i_{\widehat{Q}}} u_{\I^i_{\widehat{Q}}}(t) + w_{\I^i_{\widehat{Q}}}(t);~~
   u_{\I^i_{\widehat{Q}}} (t) = K_{\I^i_{\widehat{Q}}} x_{\I^i_{\widehat{Q}}}(t). 
\end{align}
According to the Bellman equation in an infinite horizon average cost MDP~(\cite{bertsekas2007dynamic}), $\widehat{Q}_i(x_{\I^i_{\widehat{Q}}}(t), u_{\I^i_{\widehat{Q}}}(t))$, $\forall i\in \mc{V}$, corresponding to the global policy {$\pi$} satisfies the fixed point equation
\begin{align}
    \lambda + \widehat{Q}_i(x_{\I^i_{\widehat{Q}}}(t), u_{\I^i_{\widehat{Q}}}(t)) = \sum_{j\in \I^i_{\text{GD}}} c_j(x_{\I^i_{C}}(t), u_{\I^i_{C}}(t)) + \mathbb{E}\big[\widehat{Q}_i(x_{\I^i_{\widehat{Q}}}(t+1), K_{\I^i_{\widehat{Q}}} x_{\I^i_{\widehat{Q}}}(t+1))\big],
    \label{eq:bellman}
\end{align}
where $\lambda \in \R$ is a free parameter to satisfy the fixed point equation. Assuming a \textit{linear architecture}, $\widehat{Q}_i(x_{\I^i_{\widehat{Q}}}, u_{\I^i_{\widehat{Q}}})$ is parameterized as $\widehat{Q}_i(x_{\I^i_{\widehat{Q}}}, u_{\I^i_{\widehat{Q}}}) = \hat{q}_i^\T \phi_i(x_{\I^i_{\widehat{Q}}}, u_{\I^i_{\widehat{Q}}}),$ where $\phi_i(\cdot)$ is some known (possibly nonlinear) basis function of the state and control and $\hat{q}_i$ are unknown parameters. 
\begin{prop}
Consider the cooperative LQR problem in~\eqref{eq:nmasopt}. For any $i\in \mc{V}$, if
    $\hat{q}_i = \text{svec}\lrp{\widehat{Q}_i},\\ \lambda = \left\langle\widehat{Q}_i, \sigma^2_w\begin{bmatrix}
    \mb{I}\\ K_{\I^i_{\widehat{Q}}}
\end{bmatrix}\begin{bmatrix}
    \mb{I}\\ K_{\I^i_{\widehat{Q}}}
\end{bmatrix}^\T\right\rangle,~\phi_i(x_{\I^i_{\widehat{Q}}}, u_{\I^i_{\widehat{Q}}}) = \text{svec}\lrp{   \begin{bmatrix}
        x_{\I^i_{\widehat{Q}}}\\
        u_{\I^i_{\widehat{Q}}}
    \end{bmatrix}\begin{bmatrix}
        x_{\I^i_{\widehat{Q}}}\\
        u_{\I^i_{\widehat{Q}}}
    \end{bmatrix}^\T },\text{ and}$\\
   $\widehat{Q}_i = \begin{bmatrix}S_{\I^i_{\widehat{Q}}}&0\\0&R_{\I^i_{\widehat{Q}}}\end{bmatrix} + \begin{bmatrix}A^\T_{\I^i_{\widehat{Q}}}\\B^\T_{\I^i_{\widehat{Q}}}\end{bmatrix}  \mathcal{L}\lrp{A_{\I^i_{\widehat{Q}}}\!+\! B_{\I^i_{\widehat{Q}}} K_{\I^i_{\widehat{Q}}}, S_{\I^i_{\widehat{Q}}}\! +\! K^\T_{\I^i_{\widehat{Q}}} R_{\I^i_{\widehat{Q}}} K_{\I^i_{\widehat{Q}}}}\lrb{A_{\I^i_{\widehat{Q}}}~~B_{\I^i_{\widehat{Q}}}},$
then the linear parameterization $~\widehat{Q}_i(x_{\I^i_{\widehat{Q}}}, u_{\I^i_{\widehat{Q}}}) = \hat{q}_i^\T \phi_i(x_{\I^i_{\widehat{Q}}}, u_{\I^i_{\widehat{Q}}})$ satisfies~\eqref{eq:bellman}.\label{prop:linarch}
\end{prop}

Denote $\phi_t = \text{svec}\left(
    \begin{bmatrix}
        x_{\I^i_{\widehat{Q}}}(t)\\
        u_{\I^i_{\widehat{Q}}}(t)
    \end{bmatrix} \begin{bmatrix}
        x_{\I^i_{\widehat{Q}}}(t)\\
        u_{\I^i_{\widehat{Q}}}(t)
\end{bmatrix}^\intercal
    \right),$
$\psi_t = \text{svec}\left(
    \begin{bmatrix}
        x_{\I^i_{\widehat{Q}}}(t)\\
        K_{\I^i_{\widehat{Q}}} x_{\I^i_{\widehat{Q}}}(t)
    \end{bmatrix} \begin{bmatrix}
        x_{\I^i_{\widehat{Q}}}(t)\\
        K_{\I^i_{\widehat{Q}}} x_{\I^i_{\widehat{Q}}}(t)
\end{bmatrix}^\intercal
    \right),$ $f = \text{svec}\left(\sigma_w^2 \begin{bmatrix}
    \mb{I}\\ K_{\I^i_{\widehat{Q}}}
\end{bmatrix}\begin{bmatrix}
    \mb{I}\\ K_{\I^i_{\widehat{Q}}}
\end{bmatrix}^\T\right),$ $\xi_t = \mathbb{E}\left[\text{svec}\left(
    \begin{bmatrix}
        x_{\I^i_{\widehat{Q}}}(t\!+\!1)\\
        u_{\I^i_{\widehat{Q}}}(t\!+\!1)
    \end{bmatrix} \begin{bmatrix}
        x_{\I^i_{\widehat{Q}}}(t\!+\!1)\\
        u_{\I^i_{\widehat{Q}}}(t\!+\!1)
\end{bmatrix}^\intercal
    \right)\right]$, and $\hat{c}^i_t =$ \\$ \sum_{j\in \I^i_{\text{GD}}} c_j(x_{\I^i_{C}}(t), u_{\I^i_{C}}(t))$ . With a slight abuse of notation, we use $\widehat{Q}_i$ to denote both the Q-function and the matrix parameterizing it.   Then,~\eqref{eq:bellman} can be rewritten as
\begin{align}
  & \hat{c}^i_t = \lambda + (\phi_t - \xi_t)\text{svec}(\widehat{Q}_i).
    \label{eq:linarchbellman}
\end{align}
Since LSTDQ is an off-policy method, $\forall~i\in \mc{V}$, we prescribe an arbitrary control law $u_{\I^i_{\widehat{Q}}}(t) = K^\text{play}_{\I^i_{\widehat{Q}}}x_{\I^i_{\widehat{Q}}}(t) + \eta_{\I^i_{\widehat{Q}}}(t),$  to generate a single trajectory $\{x_{\I^i_{\widehat{Q}}}(t),u_{\I^i_{\widehat{Q}}}(t),x_{\I^i_{\widehat{Q}}}(t\!+\!1)\}_{t=1}^T$, where $K^\text{play}_{\I^i_{\widehat{Q}}}$ may be different from $K_{\I^i_{\widehat{Q}}}$, and $\eta_{\I^i_{\widehat{Q}}}(t) \sim \mc{N}\lrp{\mbf{0}, \sigma^2_\eta \mb{I}_{n_u|\I^i_{\widehat{Q}}|}}$ is a sufficiently exciting exploration noise for learning. 
We utilize the version of LSPI in~\cite{krauth2019finite}, where new trajectory of samples are collected every iteration.
Then~\eqref{eq:linarchbellman} can be expressed in matrix form as
\begin{align}
    \mathbf{\hat{c}}_i = (\mathbf{\Phi} - \mathbf{\Xi} + \mathbf{F})\hat{q}^{\text{true}}_i,
    \label{eq:eiv_ls}
\end{align}
where
$\hat{q}^{\text{true}}_i = \text{svec}(\widehat{Q}^{\text{true}}_i)$, $\mathbf{\Phi}^\T = [ \phi_1, \phi_2,\cdots, \phi_T]
$, $\mathbf{\Xi}^\T =[ \xi_1, \xi_2, \cdots, \xi_T]$, $\mathbf{\hat{c}}^\T_i = [ \hat{c}^i_1,\hat{c}^i_2, \cdots,\hat{c}^i_T]$, $\mathbf{F}^\T = [f_1, f_2,\cdots, f_T]$,  and $\mathbf{\Psi}^\T_+ = [\psi_2,\psi_3,\cdots, \psi_{T+1}]$.
A least-squares solution to $\hat{q}^{\text{true}}_i$ in the error-in-variables problem~\eqref{eq:eiv_ls} is given by 
\begin{align}
    \hat{q}_i = 
(\mathbf{\Phi}^\intercal (\mathbf{\Phi} - \mathbf{\Psi}_+ + \mathbf{F}))^{-1} \mathbf{\Phi}^\intercal \mathbf{\hat{c}}_i.
\label{eq:eiv_lse}
\end{align}
To ensure that the least-squares estimate smat$(\hat{q}_i)$ is positive semi-definite, we perform a Euclidean projection onto the set of symmetric matrices lower bounded by $\zeta\cdot \mb{I}$ as in~\cite{krauth2019finite}.

\noindent
\textbf{Policy improvement}
The policy improvement step in the LSPI algorithms in the unstructured single agent setting have a closed-form update. However, in our multi-agent setting, a closed-form update is not possible due to the observation constraints of the agents in $\mathcal{G}_O$. Instead, agent $i\in \mc{V}$ updates its policy $K_i$ using~\eqref{eq:strpgt} as 
\begin{align}
    K_i \leftarrow K_i - \alpha \mathbb{E}\left[ \nabla_{K_i} u_i \cdot \nabla_{u_i} \widehat{Q}_i(x_{\I^i_{\widehat{Q}}},u_{\I^i_{\widehat{Q}}})\right]=K_i - 2\alpha \mathbb{E}\left[ \mb{J}_i\widehat{Q}_i\begin{bmatrix}
        x_{\I^i_{\widehat{Q}}}\\u_{\I^i_{\widehat{Q}}}
    \end{bmatrix}  x_{\I^i_O}^\T\right],~\forall~i\in\mc{V},\label{eq:detpgtexplicit}
\end{align}
where $\alpha$ is the learning rateand $\mb{J}_i \in \R^{n_u \times (n_x + n_u)|\I^i_{\widehat{Q}}|}$ is a matrix where the sub-matrix corresponding to $u_i$ is $\mb{I}_{n_u}$ and zero otherwise.

\section{Theoretical Analysis}\label{sec:theory}
We present our main result on the sample complexity of the VD set decomposition for~\eqref{eq:nmasopt}. Recall from~\cite{krauth2019finite} that a square matrix $L$ is $(\tau,\rho)$-stable if $\forall~k\in \mb{Z}_{\geq 0},$ $\norm{L^k} \leq \tau \rho^k$, where $\tau\geq 1$ and $\rho\in(0,1)$.
Let $n^i_{\widehat{x}} = n_x |\I^i_{\widehat{Q}}|,$ $n^i_{\widehat{u}} = n_u |\I^i_{\widehat{Q}}|$. 
Theorem~\ref{thm:undecsampcomp} below states the sample complexity and the estimation error of the Q-function parameter for the \textit{direct} case.
\begin{thm}
Consider $\delta\in (0,1).$  Let the initial global state and the global control (during sample generation) $\forall~t$ satisfy $x(0) \sim \mc{N}\lrp{x_0, \Sigma_0},$ $u(t) = K^{\text{play}} x(t) \!+\! \eta_t,~\eta(t)\sim \mc{N}(\mbf{0}_{Nn_u},\sigma^2_\eta \mb{I}_{Nn_u}),$ and $\sigma_\eta \leq \sigma_w$. For each $i\in \mc{V},$ let $K^{\text{play}}_{\I^i_{\widehat{Q}}},~K_{\I^i_{\widehat{Q}}}$ stabilize $(A_{\I^i_{\widehat{Q}}},B_{\I^i_{\widehat{Q}}}).$ Assume that $A_{\I^i_{\widehat{Q}}} + B_{\I^i_{\widehat{Q}}}K_{\I^i_{\widehat{Q}}}$ and $A_{\I^i_{\widehat{Q}}} + B_{\I^i_{\widehat{Q}}}K^{\text{play}}_{\I^i_{\widehat{Q}}}$ are $(\tau,\rho)$-stable. 
Let $\mathfrak{P}_\infty = \mc{L}\lrp{A_{\I^i_{\widehat{Q}}}+B_{\I^i_{\widehat{Q}}}K_{\I^i_{\widehat{Q}}}, \sigma^2_w \mb{I}_{n^i_{\widehat{x}}} + \sigma^2_\eta B_{\I^i_{\widehat{Q}}} B^\T_{\I^i_{\widehat{Q}}}}$ and $\bar{\sigma}_i = \sqrt{\tau^2 \rho^{4}||\Sigma^{\widehat{x}}_0|| + ||\mathfrak{P}_\infty|| + \sigma^2_w + \sigma^2_\eta||B_{\I^i_{\widehat{Q}}}||^2}.$ 
Suppose that $T$ satisfies
$$T \geq \tilde{O}(1) \text{max}\bigg\{(n^i_{\widehat{x}} + n^i_{\widehat{u}})^2, \dfrac{( n^i_{\widehat{x}})^2(n^i_{\widehat{x}} + n^i_{\widehat{u}})^2 ||\widehat{K}^{\text{play}}_{\I^i_{\widehat{Q}}}||^4_+}{\sigma^4_\eta}\sigma^2_w \bar{\sigma}^2_i\dfrac{\tau^4||K_{\I^i_{\widehat{Q}}}||^8_+ {{(||A_{\I^i_{\widehat{Q}}}||^2 + ||B_{\I^i_{\widehat{Q}}}||^2)^2}}}{\rho^4(1-\rho^2)^2}\bigg\}.$$
Then with probability at least $1-\delta$, we have 
$$||\hat{q}^{\text{true}}_i -\hat{q}^{\text{direct}}_i|| \leq  \dfrac{\tilde{O}(1) (n^i_{\widehat{x}} + n^i_{\widehat{u}})||K^{\text{play}}_{\I^i_{\widehat{Q}}}||^2_+ }{\sigma^2_\eta
     \sqrt{T}}\sigma_w \bar{\sigma}_i||\widehat{Q}^{\text{true}}_i||_F \dfrac{\tau^2||K_{\I^i_{\widehat{Q}}}||^4_+ {{(||A_{\I^i_{\widehat{Q}}}||^2 + ||B_{\I^i_{\widehat{Q}}}||^2)}}}{\rho^2(1-\rho^2)},$$
     where $\tilde{O}(1)$ hides $\text{ polylog}\left(\dfrac{T}{\delta},~\dfrac{1}{\sigma^4_\eta},~\tau,~n^i_{\widehat{x}},~||\Sigma_0||,~||K^{\text{play}}_{\I^i_{\widehat{Q}}}||,~||\mathfrak{P}_\infty|| \right).$
    \label{thm:undecsampcomp} 
\end{thm}
To interpret the result, observe that achieving an $\epsilon-$close estimate of $\hat{q}^{\text{true}}_i$ requires at most $$T \leq {\tilde{O}(1) }\text{max}\left(\dfrac{ W_i^2 (n^i_{\widehat{x}} + n^i_{\widehat{u}})^3}{ \sigma_\eta^4 \epsilon^2}||\widehat{Q}^{\text{true}}_i||^2, \dfrac{ W_i^2 (n^i_{\widehat{x}})^2(n^i_{\widehat{x}} + n^i_{\widehat{u}})^2}{ \sigma_\eta^4}\right)~\text{samples,}$$ where $W_i = ||K^{\text{play}}_{\I^i_{\widehat{Q}}}||^2_+ \sigma_w \bar{\sigma}_i \dfrac{\tau^2||K_{\I^i_{\widehat{Q}}}||^4_+ {{(||A_{\I^i_{\widehat{Q}}}||^2 + ||B_{\I^i_{\widehat{Q}}}||^2)}}}{\rho^2(1-\rho^2)}.$ We focus on the scaling w.r.t. the state and control dimensions $(n^i_{\widehat{x}}, n^i_{\widehat{u}})$. For agent $i\in \mc{V}$, the \textit{direct} decomposition based Algorithm~\ref{alg:malspi} is more sample efficient than the centralized case (whose state dimension is $Nn_x$ and control dimension is $Nn_u$) if $\I^i_{\widehat{Q}} \subset \mc{V}$, i.e., $|\I^i_{\widehat{Q}}|<N$. 
Under similar pre-conditions to Theorem~\ref{thm:undecsampcomp}, we prove in Appendix~\ref{sec:decsampcomp} that in the \textit{indirect} case, achieving an $\epsilon-$close estimate of ${q}^{\text{true}}_i$ requires at most $$ T \leq {\tilde{O}(1)} \underset{j\in \I^i_{\text{GD}}}{\text{ max }} \left(\text{ max }\left(\dfrac{ W_j^2 (n^j_{x} + n^j_{u})^3}{ \sigma_\eta^4 (w^i_j)^2 \epsilon^2}||Q_j||^2, \dfrac{W_j^2 (n^j_{x})^2(n^j_{x} + n^j_{u})^2}{ \sigma_\eta^4}\right)\right)~\text{samples},$$ where $W_j = ||K^{\text{play}}_{\I^i_{{Q}}}||^2_+ \sigma_w \bar{\sigma}_j \dfrac{\tau^2||K_{\I^i_{{Q}}}||^4_+ {{(||A_{\I^i_{{Q}}}||^2 + ||B_{\I^i_{{Q}}}||^2)}}}{\rho^2(1-\rho^2)}$, and $w^i_1,\cdots,w^i_{|\I^i_{\text{GD}}|} \in \R_+$ satisfy {$\sum_{j=1}^{|\I^i_{\text{GD}}|} w^i_j = 1.$}  The user-defined weights {$w^i_j$} can be construed as the relative importance of the estimation accuracies of $q_j$, $\forall~j\in \I^i_{GD}.$ By a specific choice of {$w^i_j$} and the analysis in Appendix~\ref{app:remark}, we show that the worst-case sample complexity of the \textit{indirect} decomposition-based Algorithm~\ref{alg:malspi} is equal to that of the \textit{direct} case and strictly better if  $\I^j_{{Q}} \subset \I^i_{\widehat{Q}},~\forall~j\in\I^i_{\text{GD}}.$ The necessary and sufficient graphical conditions to ensure strictly better sample efficiency of direct and indirect methods are derived in Lemma~\ref{lm:gcondn} (Appendix~\ref{sec:gcondn}).

\section{Simulations}\label{sec:sim} 

Consider $N$ agents, whose dynamics are given in~\cite[Example 1]{krauth2019finite}.  
We investigate the performance of Algorithm~\ref{alg:malspi} in two examples. \textbf{Example 1}:  We prescribe the inter-agent couplings $\mc{E}_R = \lrs{(i,i)|i\in \mc{V}},~\mc{E}_S = \mc{E}_O =\lrs{(j,j-1),(j,j+1)| j = 2k+1 \leq N, k\in \mb{N}} \cup \lrs{(1,N)}\cup \lrs{(N,1), \text{ if } N = 2k+1 \text{ for some } k\in \mb{N}} \cup \lrs{(i,i)|i\in \mc{V}},$ which yield $\mc{E}_{Q} = \mc{E}_O$. \textbf{Example 2}: We prescribe the inter-agent couplings $\mc{E}_S = \lrs{(i,i)|i\in \mc{V}},~\mc{E}_O = \mc{E}_R=\lrs{(1,j)| j \in \mc{V}} \cup \lrs{(i,i)|i\in \mc{V}},$ corresponding to a leader-follower network which yields $\mc{E}_{Q} = \mc{E}_R.$ Example 1 demonstrates the case where $\I^i_Q \subset \I^i_{\widehat{Q}} \subset \mc{V}$ holds $\forall~i\in \mc{V}$ with $\underset{i}{\text{max}} (|\I^i_{\widehat{Q}}| - |\I^i_Q|) = 4$ for any $N$. In contrast, Example 2 highlights the case where
for the leader $l\in\mc{V}$, it always holds that $\I^l_Q \subset\I^l_{\widehat{Q}} = \mc{V}$.
We compare the direct and indirect decompositions with a `centralized' (CTDE) MALSPI baseline in which each agent learns $Q_i(x(t),u(t))$ in the policy evaluation step while the individual control is still subject to $\mc{G}_O$. 
In addition, we examine the `undecomposed direct' method, in which each agent learns $\widehat{Q}_i(x(t),u(t))$ in the policy evaluation step, to study the effect of decomposition of the Q-function on the performance of the learned controller.
The decomposition of the above baselines and the simulation parameters are summarized in Table~\ref{tab:comptab},~\ref{table:simparam1} (Appendix~\ref{sec:simparam}) respectively.  
We simulate $N = 8$ agents with an evaluation trajectory length of $T_{\text{eval}}= 500$ steps. Fig.~\ref{fig:8agentrew1}a and~\ref{fig:8agentrew1}b show the comparison of the total average cost for $20$ MC simulations using different Q-function architectures for Example 1 and 2, respectively. In both cases, we observe that the indirect method has the fastest rate of convergence and lowest average cost followed by the direct, undecomposed direct, and centralized methods. Given a sufficient number of samples, the centralized and undecomposed direct methods perform comparably to the direct and indirect methods. This empirically corroborates the strictly better sample efficiency of the Q-function decomposition as discussed in Section~\ref{sec:theory}. The difference in the performance of the direct and indirect methods in the low-sample regime is more pronounced in Example 2 than in Example 1. This is attributed to the comparable sparsity of $\I^i_Q$ and $\I^i_{\widehat{Q}}$ (since $\underset{i}{\text{max}} (|\I^i_{\widehat{Q}}| - |\I^i_Q|) = 4$) in Example 1,  resulting in comparable sample efficiency for both methods. However, in Example 2, $\I^1_{\widehat{Q}} =\mc{V}$ for agent 1 (the leader), rendering the direct method to be less sample efficient compared to the indirect method. Table~\ref{table:avgtime} (Appendix~\ref{sec:simparam}) summarizes the average computational time per iteration of Algorithm~\ref{alg:malspi} for $N = 8,~20,~40$;~$T = 500$ steps in Example 1, 2. We observe that the computational time for the centralized and undecomposed direct methods scales exponentially with the number of agents whereas the time for the direct and indirect decomposition scales only with the number of agents in the $\I^i_Q$ and $\I^i_{\widehat{Q}}$, respectively. This corroborates the computational savings and scalability achieved by the Q-function decomposition proposed in Section~\ref{sec:decompQ}.

\begin{figure}[htpb]
    \centering
    % First group of images
    \begin{subfigure}
        \centering
        \includegraphics[width=0.22\textwidth]{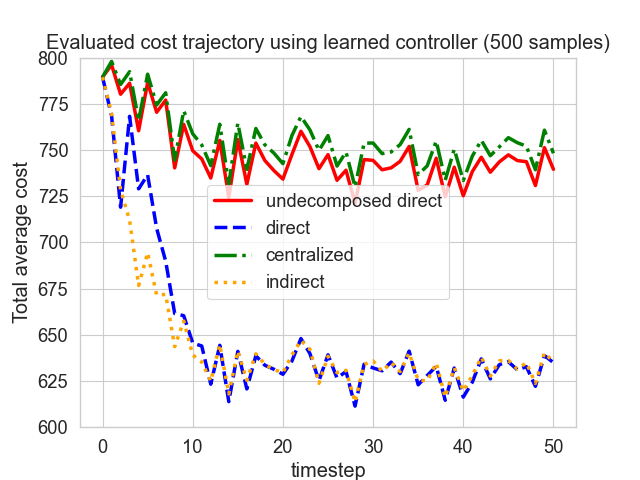}
        \includegraphics[width=0.22\textwidth]{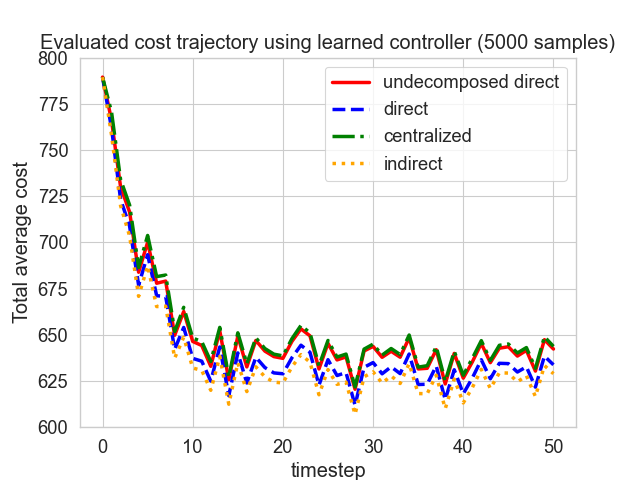}
        % \caption{Example 1}
        \label{fig:suba}
    \end{subfigure}
    % Second group of images
    \begin{subfigure}
        \centering
        \includegraphics[width=0.22\textwidth]{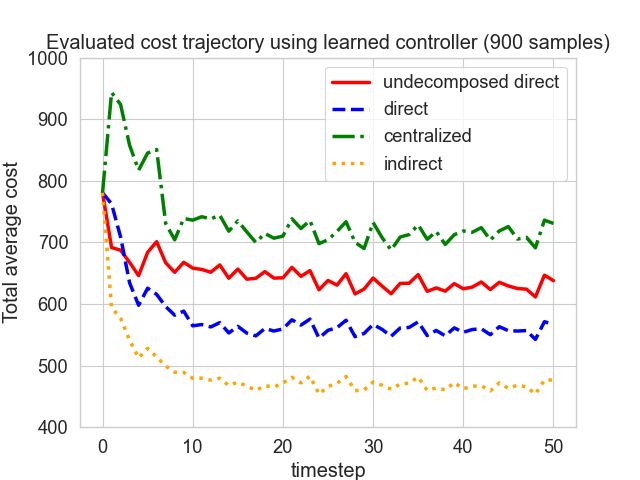}
        \includegraphics[width=0.22\textwidth]{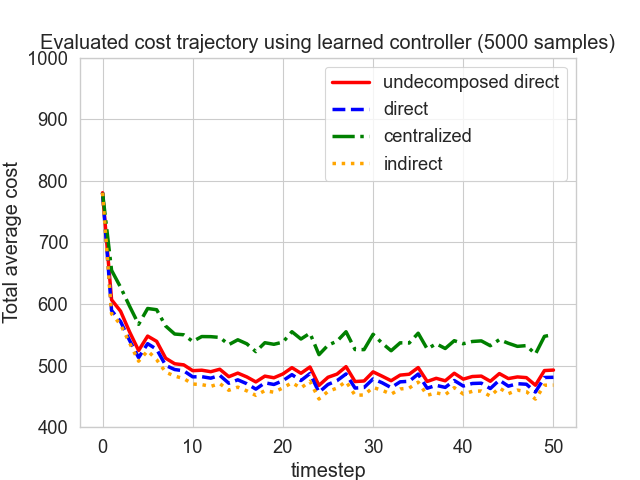}
        % \caption{Example 2}
        \label{fig:b}
    \end{subfigure}
    \caption{Comparison of the total average cost for Example 1 (a) and Example 2 (b) using direct, indirect, undecomposed direct, and centralized Q-function architectures in Algorithm~\ref{alg:malspi}.}
    \label{fig:8agentrew1}
\end{figure}

\section{Conclusion}\label{sec:conc}
We develop a systematic approach to leverage inter-agent coupling information and perform exact decompositions of the individual Q-function and the gradient of the global Q-function. Based on the decompositions, we introduce a cooperative deterministic policy gradient theorem and a cooperative MALSPI algorithm. We establish the theoretical sample and error guarantees for the obtained decomposition and provide necessary and sufficient graphical conditions for better sample efficiency of the proposed decomposition. {We empirically validate the improved sample and computational efficiency using two numerical examples.} Our future work will investigate the effect of approximate decompositions on the Q-function estimation and the convergence of the  MALSPI algorithm.

\newpage
\acks{This work was partly supported by the U.S. DEVCOM Army Research Laboratory (ARL) under Cooperative Agreement W911NF2120219 and National Science Foundation (NSF) award \#2212582. Any opinion, findings, and conclusions or recommendations expressed in this material are those of the authors and should not be interpreted as representing the official policies, either expressed or implied, of ARL, NSF or the U.S. Government.}
\bibliography{references}

\newpage
\begin{center}
    \Large{\textbf{Appendix}}
\end{center}

\appendix
\section{Proof of Lemma~\ref{lm:rtc}}\label{sec:rtc}
\begin{proof}
We prove a stronger version of the lemma that holds irrespective of the linear dynamics and quadratic cost assumption.
        For some $i,j\in \mathcal{V}$, let $j \in \mathcal{I}^i_Q$. For the sake of contradiction, assume that $\exists$ a $k \in \mathcal{R}^j_{SO}$ such that $k \notin \mathcal{I}^i_Q$. By the definition of  $\mathcal{I}^i_Q$, $j \in \mathcal{I}^i_Q$ implies that 
        for some $t' \geq t$, $\exists$ a function (or composition of functions) $f:\mathcal{S}\times \mathcal{U} \rightarrow \R$ such that 
    \begin{align}
        c_i(x_{\I^i_C}(t'),u_{\I^i_C}(t')) = f(x_j(t), u_j(t), \bigcup_{g\in \I^i_Q\setminus{j}} \{x_g(\cdot), u_g(\cdot)\}).
        \label{eq:reward}
    \end{align}
    Recall that the control $u_j(t) \in \mc{U}$ depends only on its partial observation $o_j(t)$, current state $x_j(t)$, and local policy {$\pi_{j}(\cdot)$}. Therefore, $\exists$ a function $g_j: \mZ_j \rightarrow P(\mc{U}_j)$ such that 
    \begin{align}
        u_j(t) \sim g_j(o_j(t)) = g_j(\{x_m(t)\}_{m\in \I^j_O})
        \label{eq:action}
    \end{align}

    Similarly, due to the Markovian assumption for each $x_j(t)$, $\exists$ a mapping $h_j:\prod_{n\in \I^j_S} \mS_n \times \prod_{n\in \I^j_S} \mc{U}_n \rightarrow P(\mS_j)$ such that
    \begin{align}
        x_j(t) \sim h_j(\{x_n(t-1)\}_{n\in \I^j_S}, \{u_n(t-1)\}_{n\in \I^j_S}).
        \label{eq:statetransition}
    \end{align}

    Using~\eqref{eq:action} and~\eqref{eq:statetransition},~\eqref{eq:reward} can be rewritten as
    \begin{align}
        &c_i(x_{\I^i_C}(t'),u_{\I^i_C}(t')) = f(x_j(t), u_j(t),\bigcup_{g\in \I^i_Q\setminus{j}} x_g, u_g)\\
        &= f(h_j(\{x_n(t-1)\}_{n\in \I^j_S}, \{u_n(t-1)\}_{n\in \I^j_S}), g_j(\{x_m(t)\}_{m\in \I^j_O}), \bigcup_{g\in \I^i_Q\setminus{j}} \{x_g(\cdot), u_g(\cdot)\})\\
        &= f(h_j(\{x_n(t-1), u_n(t-1)\}_{n\in \I^j_S}), g_j(\{\{x_l(t-1), u_l(t-1)\}_{l\in \I^m_S}\}_{m\in \I^j_O}), \bigcup_{g\in \I^i_Q\setminus{j}} \{x_g(\cdot), u_g(\cdot)\}).
        \label{eq:recursivereward}\
    \end{align}

     On recursive expansion of~\eqref{eq:recursivereward}, it is straightforward to verify that $c_i(x_{\I^i_C}(t'),u_{\I^i_C}(t'))$ depends on $\{x_s(t''),u_s(t'')\}_{s\in \mathcal{R}^j_{SO}}$, for some $t''\leq t \leq t'$. Thus, $i \in \I^s_{\text{GD}}$ $\forall~s\in \mathcal{R}^j_{SO}$ which implies that $s \in \I^i_Q$  $\forall~s\in \mathcal{R}^j_{SO}$. But as $k\in \mathcal{R}^j_{SO}$, $k \in \I^i_Q$ which is a contradiction. Therefore, our assumption is false and hence if $j \in \I^i_Q$, then $\forall~k\in \mathcal{R}^j_{SO}$, $k\in \I^i_Q$ as required.
\end{proof}

\section{Proof of Theorem~\ref{thm:Qsetdecomp}}\label{sec:Qsetdecomp}
\begin{proof}
For the networked system,  observe that the individual cost-to-go for each agent $Q_i$ is dependent on the global state and control due to the long-term inter-agent dependencies between the agents. Recall that
\begin{align}
   &Q_i(x,u) =   c_i(x_{\I^i_C}, u_{\I^i_C}) +  \mb{E}\lrb{\sum_{t =1}^{\infty}  c_i(x_{\I^i_C}(t), u_{\I^i_C}(t))}.
   \label{eq:Q}
\end{align}
For the LTI dynamics~\eqref{eq:inddyn} and quadratic cost~\eqref{eq:indcost},~\eqref{eq:Q} can be rewritten as
\vspace{-10pt}
\begin{align}
   &Q_i(x,u)= \begin{bmatrix}
        x_{\I^i_C}(t)\\
        u_{\I^i_C}(t)
    \end{bmatrix}^\T \begin{bmatrix}
       S_i & 0\\
       0 & R_i
    \end{bmatrix} \begin{bmatrix}
        x_{\I^i_C}(t)\\
        u_{\I^i_C}(t)
    \end{bmatrix} + \mb{E}_{w(t),\eta(t)} \Bigg[\begin{bmatrix}
       x_{\I^i_C}(t+1)\\
        u_{\I^i_C}(t+1)
    \end{bmatrix}^\T \begin{bmatrix}
       S_i & 0\\
       0 & R_i
    \end{bmatrix} \begin{bmatrix}
        x_{\I^i_C}(t+1)\\
        u_{\I^i_C}(t+1)
    \end{bmatrix}\nonumber\\
    &~~+\mb{E}_{w(t+1),\eta(t+1)} \begin{bmatrix}
    x_{\I^i_C}(t+2)\\
        u_{\I^i_C}(t+2)
    \end{bmatrix}^\T \begin{bmatrix}
       S_i & 0\\
       0 & R_i
    \end{bmatrix} \begin{bmatrix}
        x_{\I^i_C}(t+2)\\
        u_{\I^i_C}(t+2)
    \end{bmatrix} + \mb{E}\lrb{\cdots}\Bigg]=\nonumber\\
    & \sum_{j,k\in\I^i_C} \Bigg[ (x_j(t))^\T S_{jk} (x_{k}(t)) +(u_{j}(t))^\T R_{jk} (u_{k}(t)) 
    +\lrb{\sigma^2_w \trace{S_i} + \sigma^2_\eta \trace{R_i}}_{j\in \I^i_C}  + \nonumber\\
    &
    \lrb{x^\T_{\I^j_S}(t)A^\T_j S_i A_j x_{\I^j_S}(t) +  u^\T_{\I^j_S}(t)B^\T_j S_i B_j u_{\I^j_S}(t) + 2  x^\T_{\I^j_S}(t)A^\T_j S_i  B_j u_{\I^j_S}(t) + x^\T_{\I^j_O}(t)K^\T_j R_i K_j x_{\I^j_O}(t)}_{j\in \I^i_C}\nonumber\\& +  \sigma^2_\eta \trace{B^\T_j S_i B_j \mb{I}_{n_u |\I^j_S|}} +2 \trace{A^\T_j S_i  B_j w_k(t) \eta_l^\T(t)}_{\substack{k\in \I^j_S\\ l \in \I^j_O}}  +\sigma^2_w \trace{A^\T_j S_i A_j \mb{I}_{n_x |\I^j_S|}} +\cdots\Bigg].
\label{eq:indQdec}
\end{align}
From~\eqref{eq:indQdec}, we observe that for time-invariant inter-agent couplings, the $Q_i(\cdot)$ for each $i\in \mc{V}$ depends on its neighbors in the cost graph which in turn depend on their neighbors in the state and observation graphs and so on. In other words, $\forall~i\in \mc{V},$ $Q_i(\cdot)$ depends on a subset of agents $\I^i_Q := \lrs{\I^i_C \cup \lrs{\mc{R}^k_{SO}}_{k\in \I^i_C}} = \lrs{\mc{R}^k_{SO}}_{k\in \I^i_C}.$ By Lemma~\ref{lm:rtc}, we have that $\I^i_Q$ is closed under $\mc{R}_{SO}$ which implies that the information of agents in $\I^i_Q$ is sufficient to exactly compute the the future costs of agent $i$. Thus, it follows that $Q_i(x(t),u(t)) = Q_i(x_{\I^i_Q}(t), u_{\I^i_Q}(t))$.
\end{proof}

\section{Proof of Theorem~\ref{thm:graddecomp}}\label{sec:graddecomp}
\begin{proof}
Recall that
\begin{align}
&Q(x,u) = \mathbb{E}_\pi\left[ \sum_{i=1}^N \sum_{t=0}^\infty c_i(x_{\mathcal{I}^i_C}(t), u_{\mathcal{I}^i_C}(t)) | x(0) \!=\! x, u(0)\! = \!u\right]\nonumber\\
   &=\mathbb{E}_\pi\bigg[ \sum_{j \in \mathcal{I}_{\text{GD}}^i} \sum_{t=0}^\infty c_j(x_{\mathcal{I}^j_C}(t), u_{\mathcal{I}^j_C}(t)) | x(0) \!=\! x, u(0)\! = \!u\bigg]\hspace{160pt}\nonumber\\
   &\hspace{60pt}+\mathbb{E}_\pi\bigg[ \sum_{j \backslash \mathcal{I}_{\text{GD}}^i} \sum_{t=0}^\infty c_j(x_{\mathcal{I}^j_C}(t), u_{\mathcal{I}^j_C}(t)) | x(0) \!=\! x, u(0)\! = \!u\bigg]\nonumber\\
   &= \sum_{j \in \mathcal{I}_{\text{GD}}^i} Q_j(x_{\I^j_Q},u_{\I^j_Q}) + \sum_{k \backslash \mathcal{I}_{\text{GD}}^i} Q_k(x_{\I^k_Q},u_{\I^k_Q}) = \widehat{Q}_i(x_{\I^j_{\widehat{Q}}},u_{\I^j_{\widehat{Q}}}) + \bar{Q}_i(x_{\I^i_{\bar{Q}}},u_{\I^i_{\bar{Q}}}),
   \label{eq:Qbarhat}
   \end{align}
where $\bar{Q}_i(x_{\I^i_{\bar{Q}}},u_{\I^i_{\bar{Q}}}) = Q(x,u) - \widehat{Q}_i(x_{\I^j_{\widehat{Q}}},u_{\I^j_{\widehat{Q}}}) = \sum_{k \backslash \mc{I}^i_{\text{GD}}} Q_k(x_{\I^k_Q},u_{\I^k_Q})$.
From Theorem~\ref{thm:Qsetdecomp},
the reward of each agent $i \in \mathcal{V}$ depends on $x_j(t)$, $u_j(t)$ $\forall$ $j \in \I^i_Q$ and $\mathcal{E}_{\text{GD}} = \mathcal{E}^\intercal_{Q}$ by the definition of $\mathcal{G}_{\text{GD}}$. Therefore, if $j \notin \mathcal{I}^i_{\text{GD}}$, then $i \notin \I^j_Q$. Hence, $\sum_{j \backslash \mathcal{I}^i_{\text{GD}}} c_j(x_{\mathcal{I}^j_C}(t), u_{\mathcal{I}^j_C}(t))$ is independent of $u_i(t)$ and thus $K_i$. It then follows that $ Q_j(\cdot)$ is independent of $K_i$, $\forall$ $j \notin \mathcal{I}^i_{\text{GD}}$, which implies
\begin{align}
    \nabla_{K_i} \  \bar{Q}_i &=  
     \nabla_{K_i}\mathbb{E}_\pi\bigg[  \sum_{j \backslash \mathcal{I}_{\text{GD}}^i} \sum_{t=0}^\infty c_j(x_{\mathcal{I}^j_C}(t), u_{\mathcal{I}^j_C}(t)) | x(0) \!=\! x, u(0)\! = \!u\bigg] \nonumber\\
     &\overset{(a)}{= }  \mathbb{E}_\pi\bigg[ \nabla_{K_i} \sum_{j \backslash \mathcal{I}_{\text{GD}}^i} \sum_{t=0}^\infty c_j(x_{\mathcal{I}^j_C}(t), u_{\mathcal{I}^j_C}(t)) | x(0) \!=\! x, u(0)\! = \!u\bigg]= 0,
  \label{eq:gradindependence}
\end{align}
 where (a) in~\eqref{eq:gradindependence} is obtained by interchanging the derivative and integral {because each $Q_j(\cdot)$ is a quadratic function (as shown in Proposition~\ref{prop:linarch}), and thus sufficiently smooth in state and control.}
Hence, the gradient of the global action value function with respect to $K_i$ is given by $ \nabla_{K_i}  Q(s,a) =  
   \nabla_{K_i}  [\widehat{Q}_i + \bar{Q}_i ]= \nabla_{K_i} \widehat{Q}_i.$
\end{proof}

\section{Proof of Proposition~\ref{prop:linarch}}\label{sec:linarch}
\begin{proof}
    From~\eqref{eq:bellman}, we have  
    \begin{align}
   \widehat{Q}_i(x_{\I^i_{\widehat{Q}}}, u_{\I^i_{\widehat{Q}}}) &=  \begin{bmatrix}
        x_{\I^i_{\widehat{Q}}}(t)\\
        u_{\I^i_{\widehat{Q}}}(t)
    \end{bmatrix} ^\T\begin{bmatrix}S_{\I^i_{\widehat{Q}}}&0\\0&R_{\I^i_{\widehat{Q}}}\end{bmatrix} \begin{bmatrix}
        x_{\I^i_{\widehat{Q}}}(t)\\
        u_{\I^i_{\widehat{Q}}}(t)
\end{bmatrix} + \mb{E}\lrb{ \widehat{Q}_i(x_{\I^i_{\widehat{Q}}}(t+1), u_{\I^i_{\widehat{Q}}}(t+1))}.
\end{align}
Then, the expected future Q-value can be rewritten as
\begin{align}
 & \mb{E}\lrb{ \widehat{Q}_i(x_{\I^i_{\widehat{Q}}}(t\!+\!1), u_{\I^i_{\widehat{Q}}}(t\!+\!1))}\nonumber\\
&=\mb{E}\lrb{\begin{bmatrix}
        x_{\I^i_{\widehat{Q}}}(t\!+\!1)\\
        u_{\I^i_{\widehat{Q}}}(t\!+\!1)
    \end{bmatrix} ^\T\begin{bmatrix}S_{\I^i_{\widehat{Q}}}&0\\0&R_{\I^i_{\widehat{Q}}}\end{bmatrix} \begin{bmatrix}
        x_{\I^i_{\widehat{Q}}}(t\!+\!1)\\
        u_{\I^i_{\widehat{Q}}}(t\!+\!1)
\end{bmatrix}} + \mb{E}\lrb{\mb{E}\lrb{ \widehat{Q}_i(x_{\I^i_{\widehat{Q}}}(t\!+\!2), u_{\I^i_{\widehat{Q}}}(t\!+\!2))}}\nonumber\\
&= \mb{E}\lrb{(A_{\I^i_{\widehat{Q}}} x_{\I^i_{\widehat{Q}}}(t) + B_{\I^i_{\widehat{Q}}} u_{\I^i_{\widehat{Q}}}(t) + w_{\I^i_{\widehat{Q}}}(t))^\T S_{\I^i_{\widehat{Q}}}(A_{\I^i_{\widehat{Q}}} x_{\I^i_{\widehat{Q}}}(t) + B_{\I^i_{\widehat{Q}}} u_{\I^i_{\widehat{Q}}}(t) + w_{\I^i_{\widehat{Q}}}(t)) } + \nonumber\\
& \mb{E}\Big[(K_{\I^i_{\widehat{Q}}}(A_{\I^i_{\widehat{Q}}} x_{\I^i_{\widehat{Q}}}(t) + B_{\I^i_{\widehat{Q}}} u_{\I^i_{\widehat{Q}}}(t) + w_{\I^i_{\widehat{Q}}}(t)) 
)^\T  R_{\I^i_{\widehat{Q}}} (K_{\I^i_{\widehat{Q}}}(A_{\I^i_{\widehat{Q}}} x_{\I^i_{\widehat{Q}}}(t) + B_{\I^i_{\widehat{Q}}} u_{\I^i_{\widehat{Q}}}(t) + w_{\I^i_{\widehat{Q}}}(t)) 
) \Big]\nonumber\\
&+ \mb{E}\lrb{\mb{E}\lrb{\widehat{Q}_i(x_{\I^i_{\widehat{Q}}}(t\!+\!2), u_{\I^i_{\widehat{Q}}}(t\!+\!2))}}\nonumber\\
&= (A_{\I^i_{\widehat{Q}}} x_{\I^i_{\widehat{Q}}}(t) + B_{\I^i_{\widehat{Q}}} u_{\I^i_{\widehat{Q}}}(t))^\T S_{\I^i_{\widehat{Q}}} (A_{\I^i_{\widehat{Q}}} x_{\I^i_{\widehat{Q}}}(t) + B_{\I^i_{\widehat{Q}}} u_{\I^i_{\widehat{Q}}}(t)) + \sigma^2_w \trace{S_{\I^i_{\widehat{Q}}}+K^\T_{\I^i_{\widehat{Q}}} R_{\I^i_{\widehat{Q}}}K_{\I^i_{\widehat{Q}}}} \nonumber\\
& +(K_{\I^i_{\widehat{Q}}}(A_{\I^i_{\widehat{Q}}} x_{\I^i_{\widehat{Q}}}(t) + B_{\I^i_{\widehat{Q}}} u_{\I^i_{\widehat{Q}}}(t)))^\T  R_{\I^i_{\widehat{Q}}} (K_{\I^i_{\widehat{Q}}}(A_{\I^i_{\widehat{Q}}} x_{\I^i_{\widehat{Q}}}(t) + B_{\I^i_{\widehat{Q}}} u_{\I^i_{\widehat{Q}}}(t))) 
\nonumber\\
&+ \mb{E}\lrb{\mb{E}\lrb{\widehat{Q}_i(x_{\I^i_{\widehat{Q}}}(t\!+\!2), u_{\I^i_{\widehat{Q}}}(t\!+\!2))}}
\end{align}
\begin{align}
&= \begin{bmatrix}
        x_{\I^i_{\widehat{Q}}}(t)\\
        u_{\I^i_{\widehat{Q}}}(t)
    \end{bmatrix} ^\T\begin{bmatrix}A^\T_{\I^i_{\widehat{Q}}}\\B^\T_{\I^i_{\widehat{Q}}}\end{bmatrix} (S_{\I^i_{\widehat{Q}}} + K^\T_{\I^i_{\widehat{Q}}} R_{\I^i_{\widehat{Q}}} K_{\I^i_{\widehat{Q}}}) \lrb{A_{\I^i_{\widehat{Q}}}~~B_{\I^i_{\widehat{Q}}}}\begin{bmatrix}
        x_{\I^i_{\widehat{Q}}}(t)\\
        u_{\I^i_{\widehat{Q}}}(t)
\end{bmatrix} + \sigma^2_w \begin{bmatrix}
    \mb{I}\\ K_{\I^i_{\widehat{Q}}}
\end{bmatrix}^\T\begin{bmatrix}S_{\I^i_{\widehat{Q}}}&0\\0&R_{\I^i_{\widehat{Q}}}\end{bmatrix} \begin{bmatrix}
    \mb{I}\\ K_{\I^i_{\widehat{Q}}}
\end{bmatrix}
\nonumber\\
& + \mb{E}\lrb{\mb{E}\lrb{\widehat{Q}_i(x_{\I^i_{\widehat{Q}}}(t\!+\!2), u_{\I^i_{\widehat{Q}}}(t\!+\!2))}}.
\label{eq:Qfunexp}
\end{align}
Recursive expansion of~\eqref{eq:Qfunexp} yields
\begin{align} 
&\widehat{Q}_i(x_{\I^i_{\widehat{Q}}}, u_{\I^i_{\widehat{Q}}}) =  \begin{bmatrix}
        x_{\I^i_{\widehat{Q}}}(t)\\
        u_{\I^i_{\widehat{Q}}}(t)
    \end{bmatrix}^\T \widehat{Q}_i\begin{bmatrix}
        x_{\I^i_{\widehat{Q}}}(t)\\
        u_{\I^i_{\widehat{Q}}}(t)
\end{bmatrix} +  \sigma^2_w \begin{bmatrix}
    \mb{I}\\ K_{\I^i_{\widehat{Q}}}
\end{bmatrix}^\T \widehat{Q}_i\begin{bmatrix}
    \mb{I}\\ K_{\I^i_{\widehat{Q}}}
\end{bmatrix},
\end{align}
where with a slight abuse of notation $$\widehat{Q}_i = \begin{bmatrix}S_{\I^i_{\widehat{Q}}}&0\\0&R_{\I^i_{\widehat{Q}}}\end{bmatrix} + \begin{bmatrix}A^\T_{\I^i_{\widehat{Q}}}\\B^\T_{\I^i_{\widehat{Q}}}\end{bmatrix}  \mathcal{L}\lrp{A_{\I^i_{\widehat{Q}}}\!+\! B_{\I^i_{\widehat{Q}}} K_{\I^i_{\widehat{Q}}}, S_{\I^i_{\widehat{Q}}}\! +\! K^\T_{\I^i_{\widehat{Q}}} R_{\I^i_{\widehat{Q}}} K_{\I^i_{\widehat{Q}}}}\lrb{A_{\I^i_{\widehat{Q}}}~~B_{\I^i_{\widehat{Q}}}},$$
and $\mc{L}(X,Y)$ is the analytical solution of the discrete time Lyapunov equation $\mc{P} = X \mc{P} X^\T + Y.$
\end{proof}

\section{Necessary and sufficient graphical conditions for strictly better sample complexity}\label{sec:gcondn}
\begin{lm}
    For every $i \in \mc{V},$ $j\in \I^i_{\text{GD}}$
    \begin{enumerate}[(a)]
        \item $\I^i_{\widehat{Q}} \subset \mc{V}$ if and only if $\exists$ $k\in \mc{V}$ such that  $\forall$ $m\in  \mc{R}^i_{(SO)^\T},$ $\forall$ $p\in \mc{R}^k_{(SO)^\T}$, $\I^m_{C^\T} \cap \I^p_{C^\T} = \emptyset$.
        \item $\I^j_{{Q}} \subset \I^i_{\widehat{Q}},$ if and only if $\exists$ $m\in  \mc{R}^i_{(SO)^\T},$ $p\in \mc{R}^k_{(SO)^\T}$, for some $k\in \mc{V}\setminus \I^j_Q$ such that $\I^m_{C^\T} \cap \I^p_{C^\T} \subset \I^i_{\text{GD}}$.
    \end{enumerate}
    \label{lm:gcondn}
\end{lm}
\begin{proof} Define $\mc{R}^i_{(SO)^\T} = \lrs{j\in \mc{V}| j\xrightarrow{\mc{E}^\T_{SO}}i}$, and $\I^i_{C^\T} = \lrs{j\in \mc{V}|(j,i)\in \mc{E}^\T_O}$.
\begin{enumerate}[(a)]
\item \textbf{Necessary condition.} Assume that $\I^i_{\widehat{Q}}\subset \mc{V}.$ Then there exists a $k\in\mc{V}$ such that $k\not\in \bigcup_{j\in\I^i_{\text{GD}}} \I^j_Q$ i.e., $k\not\in \I^j_Q,$ $\forall~j\in \I^i_{\text{GD}}$.
This implies that $\forall~j\in \I^i_{\text{GD}}$, we have $k\not\in \I^j_C$ and $k\not\in \lrs{\mc{R}^p_{SO}}_{q\in \I^j_C}$. Similarly, as $j\in \I^i_{\text{GD}}$ implies $i \in \I^j_Q$, we have that either $i \in \I^j_C$ or $i \in \lrs{\mc{R}^q_{SO}}_{q\in \I^j_C}.$ 

Consider the case where $i \in \lrs{\mc{R}^q_{SO}}_{q\in \I^j_C}.$ Suppose that there exists an $r\in \I^j_C$ for which $i \in \mc{R}^r_{SO}$. Then as $k\not\in \I^j_C$ and $k\not\in \lrs{\mc{R}^q_{SO}}_{q\in \I^j_C}$, we have $\forall~m \in \mc{R}^i_{(SO)^\T}$ and $\forall~p \in \mc{R}^k_{(SO)^\T}$, $\I^m_{C^\T} \cap \I^p_{C^\T} = \emptyset$. This is because otherwise for every $l\in \I^m_{C^\T} \cap \I^p_{C^\T}$, we obtain $l\in \I^i_{\text{GD}}$ and $k \in \I^l_Q $, implying $k \in \I^i_{\widehat{Q}}$,  which contradicts our assumption.

Alternatively, if $i \in \I^j_C$, $k\not\in \I^j_C$ and $k\not\in \lrs{\mc{R}^q_{SO}}_{q\in \I^j_C}$ imply that $\forall~p \in \mc{R}^k_{(SO)^\T}$, $\I^i_{C^\T} \cap \I^p_{C^\T} = \emptyset$. Otherwise for every $l\in \I^i_{C^\T} \cap \I^p_{C^\T}$, we obtain $l\in \I^i_{\text{GD}}$, and $k \in \I^l_Q $ implying $k \in \I^i_{\widehat{Q}}$  which contradicts our assumption.

As $i\in \mc{R}^i_{(SO)^\T}$, we have that  $\I^i_{C^\T} \cap \I^p_{C^\T} = \emptyset$ whenever $\I^m_{C^\T} \cap \I^p_{C^\T} = \emptyset$. Therefore, we conclude that if $\I^i_{\widehat{Q}}\subset \mc{V}$, then $\forall~m \in \mc{R}^i_{(SO)^\T}$, $\forall~p \in \mc{R}^k_{(SO)^\T}$, $\I^m_{C^\T} \cap \I^p_{C^\T} = \emptyset$.

\textbf{Sufficient condition.}

Consider an $i\in \mc{V}$ and assume that there exists a $k\in \mc{V}$ such that $\forall~m \in \mc{R}^i_{(SO)^\T}$, $\forall~p \in \mc{R}^k_{(SO)^\T}$, $\I^m_{C^\T} \cap \I^p_{C^\T} = \emptyset$. Consider $s \in \I^i_{\text{GD}}$, which means $i\in \I^s_Q.$ It then follows that either $i\in \I^s_C$ or $i \in \lrs{\mc{R}^n_{SO}}_{n\in \I^s_C}$.

If $i \in \I^s_C$, then as $i\in \mc{R}^i_{(SO)^\T}$, we have that $\forall~p  \in \mc{R}^k_{(SO)^\T}$, $\I^i_{C^\T} \cap \I^p_{C^\T} = \emptyset$, which results in $p \not\in \I^s_C$. This is because otherwise $s\in \I^i_{C^\T} \cap \I^p_{C^\T}$. Also, as $k\in\mc{R}^k_{(SO)^\T}$, we have $k\not\in \I^s_C$.

For any $n\in \I^s_C$ such that $i \in \mc{R}^n_{SO}$, it follows that 
$n \in \mc{R}^i_{(SO)^\T}$. Therefore, $\forall~p \in \mc{R}^k_{(SO)^\T}$, $\I^n_{C^\T} \cap \I^p_{C^\T} = \emptyset$, which means $k\not\in \mc{R}^n_{SO}$ for any $n\in \I^s_C$ such that $i \in \mc{R}^n_{SO}$. Let $n_1, n_2 \in \I^s_C$, where $n_1 \neq n_2$ such that $i \in \mc{R}^{n_1}_{SO}$ but $i \not\in \mc{R}^{n_2}_{SO}$. Then, as $n_1 \in \mc{R}^i_{(SO)^\T}$, and $n_1, n_2 \in \I^s_C$, we have $k \not\in \mc{R}^{n_2}_{SO}$. This is because otherwise $s \in \I^{m}_{C^\T} \cap \I^{p}_{C^\T}$ for $m = n_1$ and $p = n_2$, which contradicts our assumption. Therefore, we conclude that $\forall~n\in \I^s_C,$ $k\not\in \mc{R}^n_{SO}$. 

It follows from $k\not\in \I^s_C$ and $k\not\in \lrs{\mc{R}^n_{SO}}_{n\in \I^s_C}$ that $k \not\in \I^s_Q$ $\forall~s\in \I^i_{\text{GD}}$, i.e., $k\not\in \I^i_{\widehat{Q}}$. As $k \in \mc{V}\setminus \I^i_{\widehat{Q}}$,  $\mc{V}\setminus \I^i_{\widehat{Q}}$ is non-empty, i.e., $\I^i_{\widehat{Q}}\subset \mc{V}$.

\item \textbf{Necessary condition.} Consider an $i\in \mc{V}$ and assume that there exists a $j\in \I^i_{\text{GD}}$, such that $\I^j_Q\subset \I^i_{\widehat{Q}}.$ This implies that $\exists~k\in\I^i_{\widehat{Q}}$ such that $k\not\in \I^j_Q$, and  $k \in \bigcup_{h\in \I^i_{\text{GD}}\setminus\lrs{j}}\I^h_Q.$ If $k\not\in \I^j_Q$, then by definition, $k \not\in \I^j_C,$ and $k \not\in \lrs{\mc{R}^l_{SO}}_{l\in \I^j_C}$.  
But, $k \in \bigcup_{h\in \I^i_{\text{GD}}\setminus\lrs{j}}\I^h_Q$ implies that $\exists~h\in \I^i_{\text{GD}}\setminus\lrs{j}$ such that either $k\in \I^h_C$ or $k\in \lrs{\mc{R}^m_{SO}}_{m\in \I^h_C}.$ 

\textbf{Case 1} Let  $k\in \I^h_C$. Then, as $h \in \I^i_{\text{GD}}$, either $i \in \I^h_C$, or $i \in \lrs{\mc{R}^l_{SO}}_{l\in \I^h_C}$.  
\begin{itemize}
    \item If $i \in \I^h_C$, then $\I^i_{C^\T}\cap \I^k_{C^\T} = \lrs{h} \neq \emptyset.$ or,
    \item If $i \in \lrs{\mc{R}^l_{SO}}_{l\in \I^j_C}$, then $\exists$ an $m\in \I^h_C \cap \mc{R}^i_{(SO)^\T}.$ Hence, $\I^m_{C^\T} \cap \I^k_{C^\T} = \lrs{h} \neq \emptyset.$
\end{itemize}

\textbf{Case 2} Let  $k\in \lrs{\mc{R}^m_{SO}}_{m\in \I^h_C}.$ Then, $\exists$ an $p\in \I^h_C \cap \mc{R}^k_{(SO)^\T},$ and as $h \in \I^i_{\text{GD}}$, either $i \in \I^h_C$, or $i \in \lrs{\mc{R}^l_{SO}}_{l\in \I^h_C}$.  
\begin{itemize}
    \item If $i \in \I^j_C$, then $\I^i_{C^\T} \cap \I^p_{C^\T} = \lrs{h} \neq \emptyset.$ or,
    \item If $i \in \lrs{\mc{R}^l_{SO}}_{l\in \I^h_C}$, then 
    $\exists$ an $m\in \I^h_C \cap \mc{R}^i_{(SO)^\T}.$ Hence, $\I^m_{C^\T} \cap \I^p_{C^\T} = \lrs{h} \neq \emptyset.$ 
\end{itemize}

Therefore, in either case we conclude that if $\I^i_Q\subset \I^i_{\widehat{Q}},$ then  $p\in \mc{R}^k_{(SO)^\T},$ $m\in  \mc{R}^i_{(SO)^\T},$ such that $\I^m_{C^\T} \cap \I^p_{C^\T} \subset \I^i_{\text{GD}}.$

\textbf{Sufficient condition.} Consider an $i\in \mc{V}$ and assume that $\exists$ $j\in \I^i_{\text{GD}}$ for which $\exists$ $k\in \mc{V}\setminus \I^j_Q$. Let $h\in \I^i_{\text{GD}}$, $m\in  \mc{R}^i_{(SO)^\T},$ and $p\in \mc{R}^k_{(SO)^\T},$ such that $h \in \I^m_{C^\T} \cap \I^p_{C^\T}.$ Hence, as $p \in \I^h_C,$ by definition $k\in \I^h_Q.$ As $h \in \I^i_{\text{GD}}$, we have that $k\in \I^i_{\widehat{Q}}$. However, $k\not \in \I^j_Q$ implies that $k\in \I^i_{\widehat{Q}}\setminus \I^i_Q$ or $\I^i_Q \subset \I^i_{\widehat{Q}}$
as required.
\end{enumerate}
\end{proof}

\section{Proof of Theorem~\ref{thm:undecsampcomp}}\label{sec:undecsampcomp}

\begin{proof}
For the analysis of the direct case, we first show that for each $i\in \mc{V}$, $||\hat{q}^{\text{true}}_i -\hat{q}^{\text{direct}}_i||$ is analogous to Lemma A.1,~\cite{krauth2019finite} in the single-agent case. For brevity, in the remainder of the proof we denote $\hat{q}^{\text{direct}}_i$ by $\hat{q}_i$.
From $\eqref{eq:eiv_lse}$,  the solution error-in-variables least squares is given by 
\begin{align}
    \hat{q}_i = (\mathbf{\Phi}^\intercal (\mathbf{\Phi} - \mathbf{\Psi}_+ + \mathbf{F}))^{-1} \mathbf{\Phi}^\intercal \mathbf{\hat{c}}_i.
\label{eq:eivlse}
\end{align}
Rearranging the terms in \eqref{eq:eivlse} yields 
\begin{align}
  &\mathbf{\Phi}^\intercal (\mathbf{\Phi} - \mathbf{\Psi}_+ + \mathbf{F})  \hat{q}_i = 
\mathbf{\Phi}^\intercal \mathbf{\hat{c}}_i \Rightarrow \bphi \hat{q}_i = \bphi (\bphi^\intercal\bphi)^{-1}\bphi^\intercal(\mathbf{\hat{c}}_i + (\bpsi_+ -  \mathbf{F}) \hat{q}_i).
\label{eq:phat}
\end{align}
Define $P_\bphi = \bphi (\bphi^\intercal\bphi)^{-1}\bphi^\intercal$ as the orthogonal projection onto the columns of $\bphi.$ Combining \eqref{eq:eiv_ls},~\eqref{eq:phat}, and using the fact that $P_\bphi \bphi = \bphi$ yields
\begin{align}
    P_\bphi(\bphi - \bxi + \mathbf{F})(\hat{q}^{\text{true}}_i - \hat{q}_i) = P_\bphi(\bxi - \bpsi_+)\hat{q}_i .
    \label{eq:errorphat}
\end{align}
The $i^{\text{th}}$ row of $\bphi - \bxi + \mathbf{F}$ can be expressed as,
\begin{align}
    &\text{svec}\bigg(
    \begin{bmatrix}
        x_{\I^i_{\widehat{Q}}}(t)\\
        u_{\I^i_{\widehat{Q}}}(t)
    \end{bmatrix} \begin{bmatrix}
        x_{\I^i_{\widehat{Q}}}(t)\\
        u_{\I^i_{\widehat{Q}}}(t)
\end{bmatrix}^\intercal - \mathbb{E}\left[\begin{bmatrix}
        x_{\I^i_{\widehat{Q}}}(t+1) \\
        K_{\I^i_{\widehat{Q}}}(x_{\I^i_{\widehat{Q}}}(t+1))
    \end{bmatrix} \begin{bmatrix}
        x_{\I^i_{\widehat{Q}}}(t+1) \\
        K_{\I^i_{\widehat{Q}}}(x_{\I^i_{\widehat{Q}}}(t+1))
\end{bmatrix}^\intercal\right]+\sigma^2_w\begin{bmatrix}
    \mb{I}\\ K_{\I^i_{\widehat{Q}}}
\end{bmatrix}\begin{bmatrix}
    \mb{I}\\ K_{\I^i_{\widehat{Q}}}
\end{bmatrix}^\T
    \bigg),\nonumber\\
&=       \text{svec}\bigg(
    \begin{bmatrix}
        x_{\I^i_{\widehat{Q}}}(t)\\
        u_{\I^i_{\widehat{Q}}}(t)
    \end{bmatrix} \begin{bmatrix}
        x_{\I^i_{\widehat{Q}}}(t)\\
        u_{\I^i_{\widehat{Q}}}(t)
\end{bmatrix}^\intercal - L \begin{bmatrix}
        x_{\I^i_{\widehat{Q}}}(t)\\
        u_{\I^i_{\widehat{Q}}}(t)
    \end{bmatrix} \begin{bmatrix}
        x_{\I^i_{\widehat{Q}}}(t)\\
        u_{\I^i_{\widehat{Q}}}(t) \end{bmatrix}^\intercal L^\intercal
    \bigg) =  (\mathbb{I} - L \otimes_s L)\phi_t,\nonumber\\
    &\text{where } L = \begin{bmatrix}
        A_{\I^i_{\widehat{Q}}}& B_{\I^i_{\widehat{Q}}} \\   K_{\I^i_{\widehat{Q}}}A_{\I^i_{\widehat{Q}}}& K_{\I^i_{\widehat{Q}}}B _{\I^i_{\widehat{Q}}}
    \end{bmatrix}.
    \label{eq:irow}
    \end{align}
Combining~\eqref{eq:errorphat},~\eqref{eq:irow} and assuming that $\bphi$ is full column rank, we obtain
\begin{align}
    &\bphi (\mathbb{I} - L \otimes_s L)^\intercal (\hat{q}^{\text{true}}_i - \hat{q}_i) = P_\bphi(\bxi - \bpsi_+)\hat{q}_i\nonumber\\
    &\Rightarrow (\mathbb{I} - L \otimes_s L)^\intercal (\hat{q}^{\text{true}}_i - \hat{q}_i) = (\bphi^\intercal\bphi)^{-1}\bphi^\intercal(\bxi - \bpsi_+)\hat{q}_i.
    \label{eq:errlse1}
\end{align}
Let $\sigma_{\text{min}}(\cdot)$ denote the minimum singular value of  a matrix. Then, we have that 
\begin{align}
    ||(\mathbb{I} - L \otimes_s L)^\intercal (\hat{q}^{\text{true}}_i - \hat{q}_i)|| &\geq \sigma_{\text{min}}(\mathbb{I} - L \otimes_s L)||\hat{q}^{\text{true}}_i - \hat{q}_i||,\label{eq:minsingdenbound}\\
    ||(\bphi^\intercal\bphi)^{-1}\bphi^\intercal(\bxi - \bpsi_+)\hat{q}_i|| &\leq \sigma_{\text{max}}((\bphi^\intercal \bphi)^{-\frac{1}{2}})||(\bphi^\intercal\bphi)^{-\frac{1}{2}}\bphi^\intercal(\bxi - \bpsi_+)\hat{q}_i||\nonumber\\
    &= \lambda_{\text{max}}((\bphi^\intercal \bphi)^{-\frac{1}{2}})||(\bphi^\intercal\bphi)^{-\frac{1}{2}}\bphi^\intercal(\bxi - \bpsi_+)\hat{q}_i||\nonumber\\
   &\hspace{-100pt}\text{ ($\because$ $\bphi^\intercal \bphi$ is symmetric and P.S.D,$(\bphi^\intercal \bphi)^{-\frac{1}{2}}$ is symmetric and P.S.D.)}\nonumber\\
    &= \dfrac{||(\bphi^\intercal\bphi)^{-\frac{1}{2}}\bphi^\intercal(\bxi - \bpsi_+)\hat{q}_i||}{{\lambda_{\text{min}}((\bphi^\intercal \bphi)^{\frac{1}{2}})}}\nonumber\\
 &= \dfrac{||(\bphi^\intercal\bphi)^{-\frac{1}{2}}\bphi^\intercal(\bxi - \bpsi_+)\hat{q}_i||}{\sqrt{\lambda_{\text{min}}(\bphi^\intercal \bphi)}}\nonumber\\
 &= \dfrac{||(\bphi^\intercal\bphi)^{-\frac{1}{2}}\bphi^\intercal(\bxi - \bpsi_+)\hat{q}_i||}{\sigma_{\text{min}}(\bphi)}
    \label{eq:minsingnumbound}
\end{align}

Combining~\eqref{eq:minsingdenbound},~\eqref{eq:minsingnumbound}, and~ \eqref{eq:errlse} yields
\begin{align}
    &\sigma_{\text{min}}(\mathbb{I} - L \otimes_s L)||\hat{q}^{\text{true}}_i - \hat{q}_i|| \leq \dfrac{||(\bphi^\intercal\bphi)^{-\frac{1}{2}}\bphi^\intercal(\bxi - \bpsi_+)\hat{q}_i||}{\sigma_{\text{min}}(\bphi)}\nonumber\\
    &\Rightarrow ||\hat{q}^{\text{true}}_i - \hat{q}_i|| \leq \dfrac{||(\bphi^\intercal\bphi)^{-\frac{1}{2}}\bphi^\intercal(\bxi - \bpsi_+)\hat{q}_i||}{\sigma_{\text{min}}(\bphi) \sigma_{\text{min}}(\mathbb{I} - L \otimes_s L)}\nonumber\\
    &\leq \dfrac{||(\bphi^\intercal\bphi)^{-\frac{1}{2}}\bphi^\intercal(\bxi - \bpsi_+)|| ||\hat{q}^{\text{true}}_i - \hat{q}_i||}{\sigma_{\text{min}}(\bphi) \sigma_{\text{min}}(\mathbb{I} - L \otimes_s L)} + \dfrac{||(\bphi^\intercal\bphi)^{-\frac{1}{2}}\bphi^\intercal(\bxi - \bpsi_+)\hat{q}^{\text{true}}_i ||}{\sigma_{\text{min}}(\bphi) \sigma_{\text{min}}(\mathbb{I} - L \otimes_s L)}\nonumber\\
    &\hspace{50pt}\text{(By triangle inequality and Cauchy-Scwartz inequality)}
\end{align}
If $\dfrac{||(\bphi^\intercal\bphi)^{-\frac{1}{2}}\bphi^\intercal(\bxi - \bpsi_+)||}{\sigma_{\text{min}}(\bphi) \sigma_{\text{min}}(\mathbb{I} - L \otimes_s L)} < \frac{1}{2},$ then
\begin{align}
    ||\hat{q}^{\text{true}}_i - \hat{q}_i|| \leq 2\dfrac{||(\bphi^\intercal\bphi)^{-\frac{1}{2}}\bphi^\intercal(\bxi - \bpsi_+)\hat{q}^{\text{true}}_i ||}{\sigma_{\text{min}}(\bphi) \sigma_{\text{min}}(\mathbb{I} - L \otimes_s L)}.
    \label{eq:errlse}
\end{align}

Observe that for each $i\in \mc{V},$ due to Lemma~\ref{lm:rtc},~\eqref{eq:errlse} is analogous to a single agent setting with state dimension $n^i_{\widehat{x}} = n_x |\I^i_{\widehat{Q}}|$, and control dimension $n^i_{\widehat{u}} = n_u |\I^i_{\widehat{Q}}|$. In the interest of space, we omit the details of the proof and provide the bound analogous to~\cite{krauth2019finite} for the direct case. Under the pre-conditions stated in Theorem~\ref{thm:undecsampcomp}, if the trajectory length $T$ satisfies
$$T \geq \tilde{O}(1) \text{max}\bigg\{(n^i_{\widehat{x}} + n^i_{\widehat{u}})^2, \dfrac{( n^i_{\widehat{x}})^2(n^i_{\widehat{x}} + n^i_{\widehat{u}})^2 ||\widehat{K}^{\text{play}}_{\I^i_{\widehat{Q}}}||^4_+}{\sigma^4_\eta}\sigma^2_w \bar{\sigma}^2_i\dfrac{\tau^4||K_{\I^i_{\widehat{Q}}}||^8_+ {{(||A_{\I^i_{\widehat{Q}}}||^2 + ||B_{\I^i_{\widehat{Q}}}||^2)^2}}}{\rho^4(1-\rho^2)^2}\bigg\},$$
then with probability at least $1-\delta$, we have that
$$||\hat{q}^{\text{true}}_i -\hat{q}_i|| \leq  \dfrac{\tilde{O}(1) (n^i_{\widehat{x}} + n^i_{\widehat{u}})||K^{\text{play}}_{\I^i_{\widehat{Q}}}||^2_+ }{\sigma^2_\eta
     \sqrt{T}}\sigma_w \bar{\sigma}_i||\widehat{Q}^{\text{true}}_i||_F \dfrac{\tau^2||K_{\I^i_{\widehat{Q}}}||^4_+ {{(||A_{\I^i_{\widehat{Q}}}||^2 + ||B_{\I^i_{\widehat{Q}}}||^2)}}}{\rho^2(1-\rho^2)},$$
     where $\tilde{O}(1)$ hides $\text{ polylog}\left(\dfrac{T}{\delta},~\dfrac{1}{\sigma^4_\eta},~\tau,~n^i_{\widehat{x}},~||\Sigma_0||,~||K^{\text{play}}_{\I^i_{\widehat{Q}}}||,~||\mathfrak{P}_\infty|| \right).$
\end{proof}

\section{Analysis of the indirect case}\label{sec:decsampcomp}
Define  $n^i_{{x}} = n_x |\I^i_{{Q}}|,$ and  $n^i_{{u}} = n_u |\I^i_{{Q}}|$.
\begin{corollary}
Consider $\delta\in (0,1).$  Let the initial global state and the global control (during sample generation) $\forall~t$ satisfy $x(0) \sim \mc{N}\lrp{x_0, \Sigma_0},$ $u(t) = K^{\text{play}} x(t) \!+\! \eta_t,~\eta(t)\sim \mc{N}(\mbf{0},\sigma^2_\eta \mb{I}_{Nn_u}),$ and $\sigma_\eta \leq \sigma_w$. For each $i\in \mc{V},$ let $K^{\text{play}}_{\I^i_{{Q}}},~K_{\I^i_{{Q}}}$ stabilize $(A_{\I^i_{{Q}}},B_{\I^i_{{Q}}}).$ Assume that $A_{\I^i_{{Q}}} + B_{\I^i_{{Q}}}K_{\I^i_{{Q}}}$ and $A_{\I^i_{{Q}}} + B_{\I^i_{{Q}}}K^{\text{play}}_{\I^i_{{Q}}}$ are $(\tau,\rho)$-stable. 
Let $\mathfrak{P}_\infty = \mc{L}\lrp{A_{\I^i_{{Q}}}+B_{\I^i_{{Q}}}K_{\I^i_{{Q}}}, \sigma^2_w \mb{I}_{n^i_x} + \sigma^2_\eta B_{\I^i_{{Q}}} B^\T_{\I^i_{{Q}}}}$ and $\bar{\sigma}_i = \sqrt{\tau^2 \rho^{4}||\Sigma^{{x}}_0|| + ||\mathfrak{P}_\infty|| + \sigma^2_w + \sigma^2_\eta||B_{\I^i_{{Q}}}||^2}.$ Further, $\forall$ $i\in \mc{V}$, let $T_i$ denote the minimum number of samples required during learning. Suppose that  
$$T_i \geq \tilde{O}(1) \text{ max} \left\{(n^i_{x} + n^i_{u})^2, \dfrac{ (n^i_{x})^2(n^i_{x} + n^i_{u})^2 ||K^{\text{play}}_{\I^j_Q}||^4_+}{\sigma^4_\eta}\sigma^2_w \bar{\sigma}^2_i\dfrac{\tau_i^4||K_{\I^j_Q}||^8_+ {{(||A_{\I^j_Q}||^2 + ||B_{\I^j_Q}||^2)^2}}}{\rho_i^4(1-\rho_i^2)^2}\right\}.$$ Then, with probability $1-\delta,$
\begin{align*}
    &\norm{ \hat{q}^{\text{true}}_i - \hat{q}^{\text{indirect}}_i } \leq\sum_{j \in \I^i_{\text{GD}}}\dfrac{\tilde{O}(1) (n^j_{x} + n^j_{u})||K^{\text{play}}_{\I^j_Q}||^2_+ }{\sigma^2_\eta
    \sqrt{T}}
     \sigma_w \bar{\sigma}_j ||Q^{\text{true}}_j||_F \dfrac{\tau_j^2||K_{\I^j_Q}||^4_+ {{(||A_{\I^j_Q}||^2 + ||B_{\I^j_Q}||^2)}}}{\rho_j^2(1-\rho_j^2)}
\end{align*}
whenever $T \geq$ max $\lrs{T_j}_{j\in \I^i_{\text{GD}}},$  where $\tilde{O}(1)$ hides $\text{polylog}\left(\dfrac{T}{\delta},~\dfrac{1}{\sigma^4_\eta},~\tau,~n_{{x}},~||\Sigma_0||,~||K^{\text{play}}_{\I^i_{{Q}}}||,~||\mathfrak{P}_\infty|| \right).$
    \label{corollary:decsampcomp}
\end{corollary}
\begin{proof}
For brevity, in the remainder of the proof denote $\phi_t = \text{svec}\left(
    \begin{bmatrix}
        x_{\I^i_{Q}}(t)\\
        u_{\I^i_{Q}}(t)
    \end{bmatrix} \begin{bmatrix}
        x_{\I^i_{Q}}(t)\\
        u_{\I^i_{Q}}(t)
\end{bmatrix}^\intercal
    \right),$\\
$\psi_t = \text{svec}\left(
    \begin{bmatrix}
        x_{\I^i_{Q}}(t)\\
        K_{\I^i_{Q}} x_{\I^i_{Q}}(t)
    \end{bmatrix} \begin{bmatrix}
        x_{\I^i_{Q}}(t)\\
        K_{\I^i_{Q}} x_{\I^i_{Q}}(t)
\end{bmatrix}^\intercal
    \right),$ $f = \text{svec}\left(\begin{bmatrix}
            \Sigma^x_{\I^i_{Q}} & \Sigma^x_{\I^i_{Q}} K^\intercal_{\I^i_{Q}}\\
            K_{\I^i_{Q}}\Sigma^x_{\I^i_{Q}}& K_{\I^i_{\widehat{Q}}}\Sigma^x_{\I^i_{Q}}K^\intercal_{\I^i_{Q}}
        \end{bmatrix}\right),$ and 
$\xi_t = \mathbb{E}\left[\text{svec}\left(
    \begin{bmatrix}
        x_{\I^i_{Q}}(t\!+\!1)\\
        u_{\I^i_{Q}}(t\!+\!1)
    \end{bmatrix} \begin{bmatrix}
        x_{\I^i_{Q}}(t\!+\!1)\\
        u_{\I^i_{Q}}(t\!+\!1)
\end{bmatrix}^\intercal
    \right)\right].$
Employing a \textit{linear architecture}, the Q-function for each agent $i$ can be expressed as
\begin{align}
  & c_i(x_{\I^i_{Q}}(t),u_{\I^i_{Q}}(t)) = \lambda + \left[\phi_t - \xi_t \right]\text{svec}(Q^{\text{true}}_i),
    \label{eq:linearbellman}
\end{align}
where $\lambda \in \R$ is a free parameter to satisfy the fixed point equation. Let $\lambda = \left\langle Q^{\text{true}}_i, \sigma^2_w\begin{bmatrix} \mb{I}_{n_x|I^i_Q|}\\  K^\intercal_{\I^i_{Q}}\end{bmatrix}\begin{bmatrix} \mb{I}_{n_x|I^i_Q|}\\  K^\intercal_{\I^i_{Q}}\end{bmatrix}^\T
         \right\rangle.$
For a single trajectory $\left\{x_{\I^i_{Q}}(t),u_{\I^i_{Q}}(t),x_{\I^i_{Q}}(t\!+\!1)\right\}_{t=1}^{T_i},$
the bellman equation for agent $i$ can be expressed in matrix form as
\begin{align}
    \mathbf{c}_i = (\mathbf{\Phi} - \mathbf{\Xi} + \mathbf{F})q^{\text{true}}_i,
    \label{eq:eivlsdec}
\end{align}
$\text{where }\mathbf{\Phi}^\T = [\phi_1,\phi_2, \cdots, \phi_{T_i}],~\mathbf{\Xi}^\T = [\xi_1, \xi_2, \cdots, \xi_{T_i}],~\mathbf{c}^\T_i = [c_i(1),c_i(2), \cdots, c_i({T_i})],~\mathbf{F}^\T = [f_1, f_2, \cdots, f_{T_i}]$.
Observe that~\eqref{eq:eivlsdec} is analogous to~\eqref{eq:eiv_ls}. Thus, using Theorem~\ref{thm:undecsampcomp}, we can conclude that if $T_i$ satisfies
\begin{align}
T_i \geq \tilde{O}(1) \text{ max} \left\{(n^i_{x} + n^i_{u})^2, \dfrac{ (n^i_{x})^2(n^i_{x} + n^i_{u})^2 ||K^{\text{play}}_{\I^j_Q}||^4_+}{\sigma^4_\eta}\sigma^2_w \bar{\sigma}^2_i\dfrac{\tau_i^4||K_{\I^j_Q}||^8_+ {{(||A_{\I^j_Q}||^2 + ||B_{\I^j_Q}||^2)^2}}}{\rho_i^4(1-\rho_i^2)^2}\right\},
    \label{eq:Tboundind1}
\end{align}
where $\bar{\sigma}_i = \sqrt{\tau^2 \rho^{4}||\Sigma^{{x}}(0)|| + ||\mathfrak{P}_\infty|| + \sigma^2_w + \sigma^2_\eta||B_{\I^i_{{Q}}}||^2}.$
Then,
\begin{align}
    ||q^{\text{true}}_i -q_i|| &\leq \dfrac{\tilde{O}(1) (n^i_{x} + n^i_{u})||K^{\text{play}}_{\I^i_Q}||^2_+ }{\sigma^2_\eta
    \sqrt{T}_i}
     \sigma_w \bar{\sigma}_i ||Q^{\text{true}}_i||_F \dfrac{\tau_i^2||K_{\I^i_Q}||^4_+ {{(||A_{\I^i_Q}||^2 + ||B_{\I^i_Q}||^2)}}}{\rho_i^2(1-\rho_i^2)}
     \label{eq:qboundind}
\end{align}
However, note that in general $\hat{q}^{\text{indirect}}_i\neq \sum_{j\in \I^i_{\text{GD}}} q_j$ as the agents might not correspond to each other if $\I^j_Q \neq \I^k_Q$, $\forall~j,k\in \I^i_{\text{GD}}$. Hence, to make the dimensions consistent, and compute the estimated local Q-function for each $j\in \I^i_{\text{GD}}$, we define a projection operator $\mc{P}^j_Q =$ blk\_diag$(P^{n_x}_{\I^j_Q, \I^i_{\widehat{Q}}}, P^{n_u}_{\I^j_Q, \I^i_{\widehat{Q}}})$, where $P^{n}_{\mathcal{S}_1,\mathcal{S}_2}$ is the projection defined in Section~\ref{sec:alg}. Then, we have that \\$\hat{q}^{\text{indirect}}_i = \sum_{j\in \I^i_{\text{GD}}} \text{svec}\lrp{(\mc{P}^j_Q)^\T \text{smat}(q_j)\mc{P}^j_Q}$, and $\hat{q}^{\text{true}}_i $ $= \sum_{j\in \I^i_{\text{GD}}} \text{svec}\lrp{(\mc{P}^j_Q)^\T \text{smat}(q^{\text{true}}_j)\mc{P}^j_Q}$.
Therefore,  the error in estimation of $\hat{q}^{\text{true}}_i$ in the indirect case can be expressed as
\begin{align}
     \norm{ \hat{q}^{\text{true}}_i - \hat{q}^{\text{indirect}}_i } &= \norm{\sum_{j\in \I^i_{\text{GD}}} \text{svec}\lrp{(\mc{P}^j_Q)^\T \text{smat}(q^{\text{true}}_j)\mc{P}^j_Q} - \sum_{j\in \I^i_{\text{GD}}} \text{svec}\lrp{(\mc{P}^j_Q)^\T \text{smat}(q_j)\mc{P}^j_Q}} \nonumber\\
     &=\norm{\sum_{j\in \I^i_{\text{GD}}} \text{svec}\lrp{(\mc{P}^j_Q)^\T (\text{smat}(q^{\text{true}}_j- q_j))\mc{P}^j_Q}}~\text{ (Due to the linearity of svec$(\cdot)$, smat$(\cdot)$)}\nonumber\\
     &= \norm{\sum_{j\in \I^i_{\text{GD}}} (q^{\text{true}}_j- q_j) } ~\text{ ($\because$ $||q_j|| = ||(\mc{P}^j_Q)^\T \text{smat}(q_j)\mc{P}^j_Q||$ $\forall~j$)}\nonumber\\
    &\leq \sum_{j \in\I^i_{\text{GD}}} \norm{q^{\text{true}}_j- q_j}~~\text{(Using triangle inequality)}.
    \label{eq:qdecerr}
\end{align}
Combining~\eqref{eq:Tboundind1},~\eqref{eq:qboundind}, and~\eqref{eq:qdecerr}, we obtain that whenever the length of trajectory (number of samples) satisfies
\begin{align}
    T \geq \text{ max }\{T_j\}_{j\in \I^i_{\text{GD}}}
\end{align}
\label{eq:Tbounddec}
then with probability $1-\delta$, we have 
\begin{align}
    &\norm{ \hat{q}^{\text{true}}_i - \hat{q}^{\text{indirect}}_i } \leq \sum_{j \in \I^i_{\text{GD}}}\dfrac{\tilde{O}(1) (n^j_{x} + n^j_{u})||K^{\text{play}}_{\I^j_Q}||^2_+ }{\sigma^2_\eta
    \sqrt{T}_j}
     \sigma_w \bar{\sigma}_j ||Q^{\text{true}}_j||_F \dfrac{\tau_j^2||K_{\I^j_Q}||^4_+ {{(||A_{\I^j_Q}||^2 + ||B_{\I^j_Q}||^2)}}}{\rho_j^2(1-\rho_j^2)},
     \label{eq:suffsamp}
\end{align}
where $\tilde{O}(1)$ hides $\text{polylog}\left(\dfrac{T}{\delta},~\dfrac{1}{\sigma^4_\eta},~\tau,~n_{{x}},~||\Sigma_0||,~||K^{\text{play}}_{\I^i_{{Q}}}||,~||\mathfrak{P}_\infty|| \right)$.
\end{proof}

\section{Remark on the sample complexity of the indirect case}\label{app:remark}
{For some $i\in \mc{V}$}, define {$w^i_1,~w^i_2,~\cdots,~w^i_{|\I^i_{\text{GD}}|} \in [0,1]$} such that {$\sum_{j=1}^{|\I^i_{\text{GD}}|} w^i_j = 1$}. Then, from~\eqref{eq:suffsamp}, observe that to achieve $\norm{ \hat{q}^{\text{true}}_i - \hat{q}^{\text{indirect}}_i }
\leq \epsilon,$ it is sufficient that every $j\in \I^i_{\text{GD}}$ satisfies $\norm{q^{\text{true}}_j- q_j} \leq {w^i_j} \epsilon$. 
From~\eqref{eq:qboundind}, we have that for any $j \in \I^i_{\text{GD}}$, $q_j$ to be $({w^i_j }\epsilon)-$optimal requires 
\begin{align*}
    T_j \leq\text{ max }\left(\dfrac{(\tilde{O}(1))^2 W_j^2 (n^j_{x} + n^j_{u})^3}{ \sigma_\eta^4 {(w^i_j)^2} \epsilon^2}||Q^{
    \text{true}}_j||^2, \dfrac{\tilde{O}(1) W_j^2 (n^j_{x})^2(n^j_{x} + n^j_{u})^2}{ \sigma_\eta^4}\right)~\text{samples.}
\end{align*}
Therefore, we conclude that to ensure $\norm{ \hat{q}^{\text{true}}_i - \hat{q}^{\text{indirect}}_i }
\leq \epsilon$, it is sufficient to have 
\begin{align*}
    T_{\text{indirect}} \leq \underset{j\in \I^i_{\text{GD}}}{\text{ max }}\left(\text{ max }\left(\dfrac{(\tilde{O}(1))^2 W_j^2 (n^j_{x} + n^j_{u})^3}{\sigma_\eta^4 {(w^i_j)^2} \epsilon^2}||Q^{
    \text{true}}_j||^2, \dfrac{\tilde{O}(1) W_j^2 (n^j_{x})^2(n^j_{x} + n^j_{u})^2}{ \sigma_\eta^4}\right)\right)~\text{samples,}
\end{align*}
 where $W_j = ||K^{\text{play}}_{\I^j_{\widehat{Q}}}||^2_+ \sigma_w \bar{\sigma}_j \dfrac{\tau^2||K_{\I^j_{\widehat{Q}}}||^4_+ {{(||A_{\I^j_{\widehat{Q}}}||^2 + ||B_{\I^j_{\widehat{Q}}}||^2)}}}{\rho^2(1-\rho^2)}$.
 
For the same $\hat{q}^{\text{true}}_i$ in the direct case, we know from Theorem~\ref{thm:undecsampcomp} that achieving $\epsilon-$optimal estimate requires at most 
\begin{align}
    T_{\text{direct}}  &\leq \text{max}\left(\dfrac{(\tilde{O}(1))^2 W_i^2 (n^i_{\widehat{x}} + n^i_{\widehat{u}})^3}{ \sigma_\eta^4 \epsilon^2}||\widehat{Q}^{\text{true}}_i||^2, \dfrac{\tilde{O}(1) W_i^2 (n^i_{\widehat{x}})^2(n^i_{\widehat{x}} + n^i_{\widehat{u}})^2}{ \sigma_\eta^4}\right)\nonumber\\
    &\leq \text{max}\left(\dfrac{(\tilde{O}(1))^2 W_i^2 (n^i_{\widehat{x}} + n^i_{\widehat{u}})^3}{ \sigma_\eta^4 \epsilon^2}\lrp{\sum_{j\in \I^i_{\text{GD}}}\norm{ Q^{\text{true}}_j}}^2, \dfrac{\tilde{O}(1) W_i^2 (n^i_{\widehat{x}})^2(n^i_{\widehat{x}} + n^i_{\widehat{u}})^2}{ \sigma_\eta^4}\right)
    \label{eq:dirT}
\end{align}
 where $W_i = ||K^{\text{play}}_{\I^i_{\widehat{Q}}}||^2_+ \sigma_w \bar{\sigma}_i \dfrac{\tau^2||K_{\I^i_{\widehat{Q}}}||^4_+ {{(||A_{\I^i_{\widehat{Q}}}||^2 + ||B_{\I^i_{\widehat{Q}}}||^2)}}}{\rho^2(1-\rho^2)}$.
We now provide an example on choosing the relative estimation weights {$w^i_j$} to achieve better sample efficiency in the indirect case compared to the direct case. Letting ${w^i_j} = ||Q^{\text{true}}_j||\lrp{\sum_{j\in \I^i_{\text{GD}}}\norm{ Q^{\text{true}}_j}}^{-1}$ yields
\begin{align}
  T_{\text{indirect}} &\leq \underset{j\in \I^i_{\text{GD}}}{\text{ max }}\left(\text{ max }\left(\dfrac{(\tilde{O}(1))^2 W_j^2 (n^j_{x} + n^j_{u})^3}{ \sigma_\eta^4 \epsilon^2}\lrp{\sum_{j\in \I^i_{\text{GD}}}\norm{ Q^{\text{true}}_j}}^2, \dfrac{\tilde{O}(1) W_j^2 (n^j_{x})^2(n^j_{x} + n^j_{u})^2}{ \sigma_\eta^4}\right)\right).
  \label{eq:indirT}
\end{align}
Note that the right hand side of~\eqref{eq:dirT} is equal to~\eqref{eq:indirT} only if there exists $j\in \I^i_{\text{GD}}$ such that $\I^j_Q = \I^i_{\widehat{Q}}$, otherwise~\eqref{eq:dirT} is strictly greater. Thus, for any $i\in \mc{V}$, $\forall~j\in \I^i_{GD},$ {$w^i_j= ||Q^{\text{true}}_j||\lrp{\sum_{j\in \I^i_{\text{GD}}}\norm{ Q^{\text{true}}_j}}^{-1}$}  ensures that the worst case sample complexity of the indirect decomposition based Algorithm~\ref{alg:malspi} is equal to the direct case and strictly better if  $\I^j_{{Q}} \subset \I^i_{\widehat{Q}},~\forall~j\in\I^i_{\text{GD}}$. We also note that the choice of {$w^i_j$} is not unique. Finding optimal weights w.r.t. the overall sample efficiency is a problem in its own interest and deferred to possible future work.

\section{Illustration of the 8-agent in the simulation examples}

\textbf{Example 1}

\begin{tikzpicture}[scale= 0.7, transform shape]
\begin{scope}
    % Define nodes and edges for the first network
    \node[latent,draw=blue!50] (A) {$1$};
    \node[latent,draw=blue!50, ,right=of A] (B) {$4$};
    \node[latent,draw=blue!50, below=of A] (C) {$2$};
    \node[latent,draw=blue!50, right=of C] (D) {$3$};
    \node[latent,draw=blue!50, right=of B] (E) {$6$};
    \node[latent,draw=blue!50 , right=of D] (G) {$5$};
    \node[latent,draw=blue!50, right=of G] (H) {$7$};
    \node[latent,draw=blue!50, right=of E] (I) {$8$};
    \node[above=of E] {$\mathcal{G}_R$};
    \path (A) edge [loop above] node (A'){}(A);
    \path (B) edge [loop above] node (B'){}(B);
    \path (C) edge [loop below] node (C'){}(C);
    \path (D) edge [loop below] node (D'){}(D);
    \path (E) edge [loop above] node (E'){}(E);
    \path (G) edge [loop below] node (G'){}(G);
    \path (H) edge [loop below] node (H'){}(H);
    \path (I) edge [loop above] node (I'){}(I);
\end{scope}
\begin{scope}[yshift=-1cm]
%Dashed vertical line
\draw[dashed] (7,-1.5) -- (7,3);
\end{scope}
\begin{scope}[xshift =8cm]
    % Define nodes and edges for the first network
    \node[latent,draw=blue!50] (A) {$1$};
    \node[latent,draw=blue!50, ,right=of A] (B) {$3$};
    \node[latent,draw=blue!50, below=of A] (C) {$2$};
    \node[latent,draw=blue!50, right=of C] (D) {$4$};
    \node[latent,draw=blue!50, right=of B] (E) {$5$};
    \node[latent,draw=blue!50 , right=of D] (G) {$6$};
    \node[latent,draw=blue!50, right=of G] (H) {$8$};
    \node[latent,draw=blue!50, right=of E] (I) {$7$};
    \node[above=of E] {$\mathcal{G}_O, \mathcal{G}_S,\mathcal{G}_Q$};
    \path (A) edge [loop above] node (A'){}(A);
    \path (B) edge [loop above] node (B'){}(B);
    \path (C) edge [loop below] node (C'){}(C);
    \path (D) edge [loop below] node (D'){}(D);
    \path (E) edge [loop above] node (E'){}(E);
    \path (G) edge [loop below] node (G'){}(G);
    \path (H) edge [loop below] node (H'){}(H);
    \path (I) edge [loop above] node (I'){}(I);
    \edge {A} {C};
    \edge {B} {C};
    \edge {B} {D};
    \edge {A} {H};
    \edge {E} {D};
    \edge {E} {G};
    \edge {I} {G};
    \edge {I} {H};
\end{scope}
\end{tikzpicture}

\textbf{Example 2}

\begin{tikzpicture}[scale= 0.7, transform shape]
\begin{scope}
    % Define nodes and edges for the first network
    \node[latent,draw=blue!50] (A) {$1$};
    \node[latent,draw=blue!50, ,right=of A] (B) {$4$};
    \node[latent,draw=blue!50, below=of A] (C) {$2$};
    \node[latent,draw=blue!50, right=of C] (D) {$3$};
    \node[latent,draw=blue!50, right=of B] (E) {$6$};
    \node[latent,draw=blue!50 , right=of D] (G) {$5$};
    \node[latent,draw=blue!50, right=of G] (H) {$7$};
    \node[latent,draw=blue!50, right=of E] (I) {$8$};
    \node[above=of E] {$\mathcal{G}_S$};
    \path (A) edge [loop above] node (A'){}(A);
    \path (B) edge [loop above] node (B'){}(B);
    \path (C) edge [loop below] node (C'){}(C);
    \path (D) edge [loop below] node (D'){}(D);
    \path (E) edge [loop above] node (E'){}(E);
    \path (G) edge [loop below] node (G'){}(G);
    \path (H) edge [loop below] node (H'){}(H);
    \path (I) edge [loop above] node (I'){}(I);
\end{scope}
\begin{scope}[yshift=-1cm]
%Dashed vertical line
\draw[dashed] (6,-1.5) -- (6,3);
\end{scope}
\begin{scope}[xshift = 12cm]
    % Define nodes and edges for the first network
    \node[latent,draw=blue!50] (A) {$1$};
    \node[latent,draw=blue!50 , below=of A] (G) {$5$};
    \node[latent,draw=blue!50, ,left=of G] (B) {$4$};
    
    \node[latent,draw=blue!50, left=of B] (D) {$3$};
    \node[latent,draw=blue!50, left=of D] (C) {$2$};

        \node[latent,draw=blue!50, right=of G] (E) {$6$};
    \node[latent,draw=blue!50, right=of E] (H) {$7$};
    \node[latent,draw=blue!50, right=of H] (I) {$8$};
    \node[above=of A] {$\mathcal{G}_O, \mathcal{G}_R, \mathcal{G}_Q$};
    \path (A) edge [loop above] node (A'){}(A);
    \path (B) edge [loop below] node (B'){}(B);
    \path (C) edge [loop below] node (C'){}(C);
    \path (D) edge [loop below] node (D'){}(D);
    \path (E) edge [loop below] node (E'){}(E);
    \path (G) edge [loop below] node (G'){}(G);
    \path (H) edge [loop below] node (H'){}(H);
    \path (I) edge [loop below] node (I'){}(I);
    \edge {A} {D};
    \edge {A} {C};
    \edge {A} {G};
    \edge {A} {H};
    \edge {A} {B};
    \edge {A} {E};
    \edge {A} {I};
\end{scope}
\end{tikzpicture}

\section{Simulation details}\label{sec:simparam}

The simulations were run on a computer with Intel(R) Core(TM) i7-4790K CPU and 32GB RAM using Pycharm IDE (Python 3.11.7).

\begin{table}[htpb]
   \centering
    \begin{tabular}{|l|c|c|c|c|c|c|}
    \hline
       Baseline  &  \multicolumn{3}{c|}{Example 1} & \multicolumn{3}{c|}{Example 2}\\
        & $N=8$& $N=20$&$N=40$&$N=8$&$N=20$&$N=40$\\
        \hline
         Centralized method&$9.39$ s& $1387.56$ s & $--$&$9.33$ s & $1391.73$ s& $--$\\
         Indirect method &$0.593$ s & $1.45$ s & $2.752$ s&$0.492$ s & $1.19$ s& $2.371$ s\\
         Direct method& $1.37$ s & $3.56$ s & $6.41$ s& $3.116$ s& $104.42$ s &$3494.72$s \\
         Undecomposed direct  &$9.41$ s&  $1396.79$ s
         & $--$& $9.4$ s &$1399.35$ s &$--$
         \\
         \hline
    \end{tabular}
    \caption{Average learning time per iteration of Algorithm~\ref{alg:malspi} for $T = 500$ steps.}
    \label{table:avgtime}
\end{table}

\begin{table}[htpb]
   \centering
    \begin{tabular}{|c|c|c|}
    \hline
       Baseline  &  Policy evaluation & Policy update\\
       \hline
         Centralized & $Q_i(x,u)$ & $\nabla_{K_i} Q_i(x,u)$\\
         Indirect &$Q_i(x_{\I^i_Q},u_{\I^i_Q})$ & $\sum_{j\in \I^i_{\text{GD}}}\nabla_{K_i} Q_j(x_{\I^j_Q},u_{\I^j_Q})$\\
         Direct & $\widehat{Q}_i(x_{\I^i_{\widehat{Q}}},u_{\I^i_{\widehat{Q}}})$ & $\nabla_{K_i} \widehat{Q}_i(x_{\I^i_{\widehat{Q}}},u_{\I^i_{\widehat{Q}}})$\\
         Undecomposed direct & $\widehat{Q}_i(x,u)$& $\nabla_{K_i}\widehat{Q}_i(x,u)$\\
         \hline
    \end{tabular}
    \caption{Comparison of the state, control dependencies in the baselines.}
    \label{tab:comptab}
\end{table}

\begin{table}[htpb]
    \centering
    \begin{tabular}{|l|c|}
    \hline
   \hspace{100pt} Parameter &  Value\\
       \hline
         Number of agents, $N$ & 8\\
         State cost matrix $S_i$, &  $s_i \otimes \mb{I}_{n_x}$\\
         Input cost matrix $R_i$ & $r_i \otimes \mb{I}_{n_x}$\\
         $s_i \in \R^{|\I^i_R|\times|\I^i_R|}$ & $s_i^{kl} =  \begin{cases}
         \frac{200}{|\I^i_R|},& \text{ if } k = l,\\
         \frac{-10}{|\I^i_R|}, & \text{otherwise}
         \end{cases}$\\
         $r_i \in \R^{|\I^i_R|\times|\I^i_R|}$&  $r_i^{kl} = \begin{cases}
             1, & \text{if } k =l,\\
             0,& \text{otherwise}
         \end{cases}$\\
         Process noise, $\sigma_w$& $ 1$\\
         Control/Exploration noise, $\sigma_\eta$ & 1\\
         Initial global stabilizing controller, $K_0$& $\mbf{0}_{Nn_u \times Nn_x}$\\
         Arbitrary global stabilizing controller (data collection), $K^{\text{play}}$ &$\mbf{0}_{Nn_u \times Nn_x}$\\         
         \hline
    \end{tabular}
    \caption{Simulation parameters for Example 1, 2.}
    \label{table:simparam1}
\end{table}
\end{document}